\documentclass[12pt]{article}
\usepackage[margin=1in]{geometry}

\usepackage[round]{natbib}

\usepackage{subfigure}
\usepackage{amsmath,amssymb}

\usepackage{amsthm} 
\allowdisplaybreaks
\usepackage{color}
\usepackage{footnote}
\usepackage{algorithm}
\usepackage{algpseudocode}
\usepackage{algorithmicx}
\usepackage{multirow} 
\usepackage{booktabs}
\usepackage{graphicx}
\usepackage{amssymb}
\usepackage{amsbsy}
\usepackage{array}
\usepackage{longtable}
\usepackage{epstopdf}
\usepackage{pbox}
\usepackage{url}
\usepackage{breqn}
\usepackage{mathrsfs}
\usepackage{multicol}
\usepackage{supertabular}
\usepackage{enumerate}
\usepackage{bbm}
\usepackage{multirow}
\usepackage{hyperref}
\usepackage{xfrac}
\usepackage{textgreek}
\usepackage[justification=centering]{caption}
\usepackage{mathtools}

\usepackage{todonotes}
\usepackage{dsfont}
\usepackage{bm}
\newtheorem{dfn}{Definition}
\newtheorem{thm}{Theorem}

\newtheorem{prp}{Proposition}

\newtheorem{lmm}{Lemma}
\newtheorem{asm}{Assumption}

\algnewcommand\algorithmicforeach{\textbf{for each}}
\algdef{S}[FOR]{ForEach}[1]{\algorithmicforeach\ #1\ \algorithmicdo}

\begin{document}
\title{\LARGE \bf Optimizing Highway Traffic Flow in Mixed Autonomy: A Multiagent Truncated Rollout Approach}

\author{Lu Liu, Chi Xie, and Xi Xiong
\thanks{This work was supported in part by NSFC Project 72371172 and Fundamental Research Funds for the Central Universities.}
\thanks{The authors are with the School of Transportation, Tongji University, Shanghai 201804, China, and also with the Key Laboratory of Road and Traffic Engineering of Ministry of Education, Tongji University, Shanghai 201804, China. (e-mail: luliu0720@tongji.edu.cn; chi.xie@tongji.edu.cn; xi\_xiong@tongji.edu.cn)}
}
\newcommand*{\QEDA}{\hfill\ensuremath{\blacksquare}}%
\date{}
\maketitle

\begin{abstract}
The development of connected and autonomous vehicles (CAVs) offers substantial opportunities to enhance traffic efficiency. However, in mixed autonomy environments where CAVs coexist with human-driven vehicles, achieving efficient coordination remains challenging due to heterogeneous driving behaviors and system complexity. To address this, this paper proposes a multiagent truncated rollout approach that optimizes CAV speed coordination to improve highway throughput while minimizing computational overhead. In this approach, a traffic density evolution equation is formulated that comprehensively accounts for the presence or absence of CAVs, and a distributed coordination control framework is established accordingly. By incorporating kinematic information from neighbor agents and employing an agent-by-agent sequential solution mechanism, our method enables explicit cooperation among CAVs. Furthermore, we introduce a truncated rollout scheme that adaptively shortens the optimization horizon based on the evaluation of control sequences. This significantly reduces time complexity, thereby improving real-time performance and scalability. Theoretical analysis provides rigorous guarantees on the input-to-state stability and performance improvement of the system. Simulations conducted on real-world bottleneck scenarios demonstrate that, in large-scale mixed traffic flows, the proposed method outperforms conventional model predictive control methods by reducing both the average travel time and overall computational time, highlighting its strong potential for practical deployment.
\end{abstract}


\section{Introduction}
Connected and autonomous vehicles (CAVs), empowered by vehicle-to-everything (V2X) communication and high-precision control, provide a technological foundation for enhancing traffic efficiency and safety \citep{LU202226partb}. However, the transition to full autonomy is gradual, and a mixed autonomy environment featuring the coexistence of CAVs and human-driven vehicles (HDVs) will persist for the foreseeable future \citep{LI2022110partb}. In such heterogeneous environments, the stochastic behaviors of HDVs significantly limit the effectiveness of CAV coordination. This challenge is particularly acute at lane-merging bottlenecks, where capacity drops trigger queue formation once vehicle density exceeds a critical threshold \citep{CHUNG200782PartB}.
Traditional ramp control methods often struggle to cope with dynamic flow fluctuations \citep{John2024TRR,JIA2025103161partb}, while infrastructure expansion is constrained by high costs and the risk of induced demand \citep{ANUPRIYA2023103726}. Although existing studies demonstrate that CAVs can improve throughput via speed control \citep{Nie2021speed}, critical limitations remain regarding interaction modeling, heterogeneous decision-making, and real-time control performance under mixed traffic conditions \citep{NTOUSAKIS2016464, Xiong2021Optimizing, CHEN2023103264}.

To bridge these gaps, we propose a distributed control framework for CAVs in mixed autonomy, aiming to enhance traffic efficiency through cooperative driving behaviors. 
We consider a three-to-two lane merging scenario as illustrated in Fig.~\ref{merging-bottleneck}, where white and blue vehicles denote HDVs and CAVs, respectively. A coordination zone is established upstream of the bottleneck. Within this zone, CAVs perceive traffic states via onboard sensors and vehicle-to-vehicle (V2V) communication, collaborating to determine acceleration or deceleration strategies that regulate the entry timing of subsequent HDVs into the bottleneck.
This CAV-centric approach not only reduces reliance on physical infrastructure but also adapts flexibly to dynamic traffic variations.
In terms of operational scope, following the approach of previous studies \citep{Vinitsky2023}, we concentrate on the longitudinal speed control of CAVs. Although explicit lane-changing decision-making is deferred, our proposed speed regulation strategy holds critical potential value: by regulating upstream density, it proactively creates favorable time gaps for merging vehicles. This approach effectively facilitates the mandatory lane changes required at the bottleneck, thereby mitigating the associated traffic friction. Furthermore, consistent with established multiagent studies \citep{NEURIPS2024_Wang,NEURIPS2024_Ding}, we assume that communication among CAVs is instantaneous. 

\begin{figure}[h]
    \centering
    \includegraphics[width=0.85\textwidth]{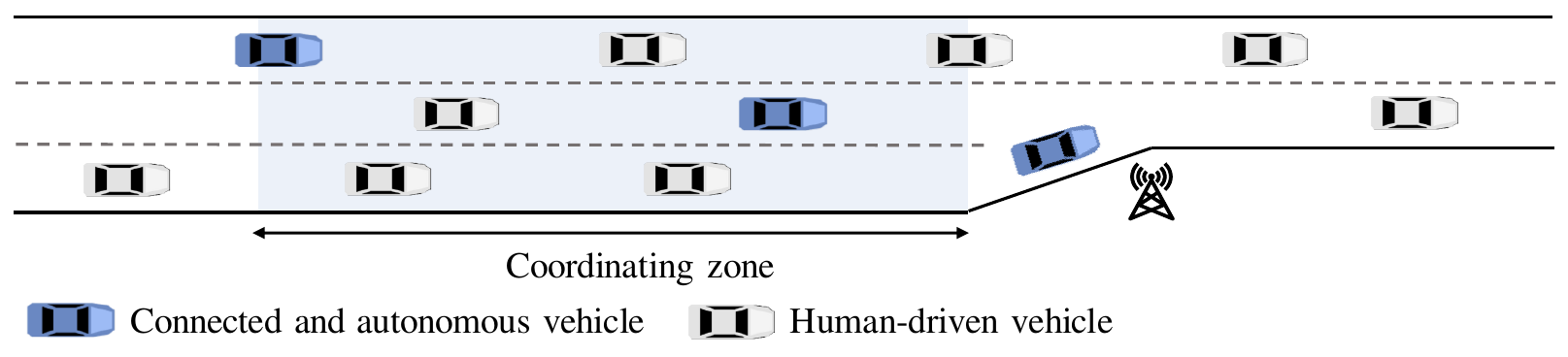}
    \caption{Illustration of CAV coordination in mixed autonomy at a highway bottleneck.}
    \label{merging-bottleneck}
\end{figure}

Most existing studies on coordinating CAVs to mitigate congestion focus on homogeneous traffic flow \citep{MA2023104266,XUE2025103209partb}. 
In mixed traffic control, microscopic models struggle with CAV coordination despite capturing individual dynamics \citep{Gao2022Optimal}, while macroscopic models achieve computational efficiency but fail to account for heterogeneous behaviors \citep{Qin2021Lighthill}.
To integrate the advantages of both microscopic and macroscopic models, \citet{Liard2023APM} proposed a PDE–ODE coupled framework that uses partial differential equations (PDEs) for density evolution and ordinary differential equations (ODEs) for CAV trajectories. Subsequent studies have demonstrated its efficacy in state estimation and flow control \citep{ZHANG2024111796,Daini2025ITS}. 
We adopt this framework to investigate cooperative CAV control in heterogeneous traffic.
However, solving coupled PDEs entails high computational complexity, imposing a significant burden on real-time applications. The cell transmission model (CTM) \citep{DAGANZO1994269} effectively addresses this by bridging the scale gap, serving as a robust discretization method for continuum models \citep{Mladen2018Traffic, piacentini2020traffic}. Building on this, \citet{QIU2023Cooperative} developed a model predictive control (MPC) approach for CAV motion planning. By minimizing objective functions over rolling horizons, the MPC implementation achieved global optimality in speed adjustments \citep{ZHANG2023199partb}.
Nevertheless, traditional centralized MPC faces dimensional challenges in complex traffic due to the exponential growth of computational costs \citep{ZHOU201969partb,SASFI2023111169}. While distributed MPC (DMPC) mitigates this by decomposing the system into parallel subsystems \citep{Goulet2022DMPC}, the reliance on lagged information exchange often leads to suboptimal solutions and complicates stability analysis \citep{LIU2022103261,ZHOU201969partb}.

The rolling optimization mechanism in traditional MPC, which employs explicit models to predict state trajectories over a finite horizon, shares conceptual similarities with the rollout mechanism in Reinforcement Learning (RL) \citep{Wang2023ITSCrollout}. In contrast to model-based MPC, RL-based multiagent control relies on agents acquiring traffic state information through environmental interaction, continuously refining policies via trial-and-error on state–action–reward outcomes \citep{Zhu2022Oper, Liu2024AMR}. 
Despite their success in traffic flow control, RL methods face significant challenges in large-scale systems, including excessive data requirements, environmental non-stationarity, and low training efficiency \citep{Bao2023ASB}. Furthermore, their computational complexity tends to scale exponentially with the system size \citep{Bertsekas2019MultiagentRA}. 
To mitigate these scalability issues, researchers have leveraged the multiagent advantage decomposition theorem to reformulate the joint policy search into a sequential decision-making process, thereby reducing complexity to be linear with the number of agents \citep{NEURIPS2021_Kuba, NEURIPS2022_Wen, JMLR-Zhong2024}. However, these methods typically necessitate a full-horizon rollout, introducing substantial computational overhead. To address this, \citet{Bhattacharya2024MultiagentTRL} proposed a multiagent truncated rollout technique that replaces full-horizon state tracking with finite-step lookahead simulations and terminal cost approximations. This approach balances control performance with computational efficiency, offering a novel paradigm for large-scale multiagent systems \citep{LI2025103154}. It is distinct from the adaptive receding horizon approach in traditional MPC \citep{Adaptive-Receding-Horizon}, which improves efficiency by adjusting the optimization variable dimensions but alters the control sequence length, potentially compromising performance.

In this study, we propose a multiagent truncated rollout approach that synergizes the predictive capability of MPC with the computational efficiency of truncated rollout. We first adopt a PDE-ODE coupled model \citep{Liard2023APM} to describe mixed traffic dynamics and derive a system-level density evolution equation via the CTM \citep{DAGANZO1994269}, capturing density variations induced by both CAVs and HDVs. 
Building on this foundation, we develop a distributed optimization framework that integrates the control information of all CAVs. In this framework, each CAV functions as an independent model predictive controller, acquiring real-time traffic states and peer trajectories through V2X communication and onboard sensors to enable accurate state prediction and optimal decision-making. 
To foster explicit cooperation and mitigate the coordination deficits of conventional DMPC, we design an agent-by-agent sequential optimization mechanism. This allows each CAV to leverage the latest control strategies of its peers, iteratively refining its actions to steer the system toward joint optimality. 
Finally, to ensure computational feasibility, we employ the truncated rollout method. By estimating the bounds of the objective function and adaptively truncating the optimization horizon, this approach significantly reduces problem dimensionality, thereby circumventing the computational burdens typical of large-scale systems.

We systematically analyze the policy improvement properties of the proposed method from the perspectives of sequential decision-making and truncated optimization, and rigorously establish the input-to-state stability (ISS) of each CAV subsystem. While strict global optimality is computationally elusive, our analysis guarantees that the control policy progressively approaches the global optimum, with the system state converging to a stable invariant set. 
Furthermore, to maximize system-level performance, we design a dynamic decision-ordering strategy and derive rigorous upper and lower bounds for the objective function. 
We also present the complete multiagent truncated rollout algorithm and conduct a time complexity analysis, demonstrating its superior computational efficiency over conventional model predictive control methods. Validation is performed using a calibrated simulation of a bottleneck segment on the Shanghai Hujin Expressway. Results indicate that, in large-scale mixed traffic, the proposed approach effectively coordinates the longitudinal dynamics of CAVs, significantly reducing both average travel time and computational overhead, thereby exhibiting strong potential for real-world deployment.

The remainder of the paper is structured as follows: 
Section~\ref{sec_model} introduces the coupled model characterizing mixed traffic flow dynamics and formulates the distributed control problem for CAVs. 
Section~\ref{sec_analysis} presents the multiagent truncated rollout approach in detail, along with its performance analysis. 
Section~\ref{sec_optimize} discusses the computational advantages of the proposed algorithm. 
Section~\ref{sec_results} evaluates the effectiveness of the method through comparative simulation experiments. 
Finally, Section~\ref{sec_conclude} concludes the paper and outlines several directions for future work. 
\section{Modeling and Formulation}
\label{sec_model}
This section introduces a fully coupled PDE-ODE model to characterize the impact of CAVs on highway traffic. By discretizing the road segment using the cell transmission model (CTM) \citep{DAGANZO1994269}, we derive the density evolution equations for individual road cells (Section~\ref{Modeling}). Furthermore, we establish a unified system density evolution equation that captures the dynamics of road cells both with and without the influence of CAVs, and formulate a distributed optimization problem in which each CAV acts as an independent actuator (Section~\ref{Problem Definition}). 

\subsection{Modeling}
\label{Modeling}
Following established research \citep{Giulia2018Traffic, M.L.2014Scalar, Piacentini2019HighwayTC, Liard2023APM, ZOU20257}, we adopt the classical Lighthill-Whitham-Richards (LWR) equation \citep{Richards1956ShockWO} to capture the evolution of traffic density. Formulated as a hyperbolic partial differential equation (PDE), this first-order macroscopic model governs traffic flow on unidirectional roads and effectively characterizes the behavior of CAVs on highways and their impact on adjacent traffic.

\begin{align}
\label{PDE}
\partial_t \rho(t,x) + \partial_x f\left(\rho(t,x)\right) = 0,
\end{align}
where $\rho(t,x) \in \left[0,R\right)$ represents the traffic density at position $x \in \mathbb{R}$ at time $t \geq 0$, $R$ is the jam density of the road. 
The flow function $f\left(\rho(t,x)\right)$ is defined as $f\left(\rho(t,x)\right) = \rho(t,x) v\left(\rho(t,x)\right)$. In this study, we adopt the Greenshields model to characterize the flow-density relationship \citep{greenshields1935study}, expressed as $v\left(\rho(t,x)\right) =V \left(1- \rho(t,x)/R \right)$, where the constant $V$ denotes the maximum vehicle speed. Consequently, this model implies that the flow $f\left(\rho(t,x)\right)$ reaches its maximum value at $\rho_*= R/2$. While triangular fundamental diagrams are widely employed in empirical studies to capture traffic characteristics \citep{SHI2021279}, the Greenshields model offers distinct advantages regarding mathematical tractability. Its continuous differentiability and strict concavity are critical for establishing the theoretical stability of the control system and ensuring the computational efficiency of the algorithm \citep{Zu2018RealtimeET, piacentini2020traffic}. Furthermore, it is important to clarify that this linear speed-density relationship is intended to characterize the aggregate equilibrium of the traffic flow, rather than the microscopic instantaneous dynamics of individual CAVs. 

The position of the CAV is represented by $y(t) \in [0, L]$, and its trajectory can be described by the following ODE \citep{M.L.2014Scalar},
\begin{align}
\label{ODE}
\partial_t y(t) = \min \Big\{ u(t), v\left(\rho\left(t,{y(t)+}\right)\right) \Big\},
\end{align}
where $L \in \mathbb{N}$ is the length of the road , $u(t) \in {U}$ is the desired velocity of the CAV, ${U}=[u_{\min},u_{\max}]$, and $\rho\left(t,{y(t)+}\right)$ denotes the density value downstream of $y(t)$. This formula indicates that the CAV can proceed at its desired speed if the downstream traffic is not congested, otherwise it will adjust to the downstream traffic speed. 

By integrating Eq.~(\ref{PDE}) and  Eq.~(\ref{ODE}), a fully coupled PDE-ODE model with multiple CAVs is developed,
\begin{subequations}
\label{PDE-ODE}
	\begin{align}
	&\partial_t \rho(t,x) + \partial_x f\left( \rho(t,x)\right) = 0, &t \textgreater 0, x \in \mathbb{R},\label{PDE-a}\\
    & \partial_t y_i(t) = \min \Big\{ u_i(t), v\left(\rho(t,y_{i}(t)+)\right) \Big\}, &t \textgreater 0, i \in \mathcal{I}(t),\\
    &f\left(\rho(t,y_i(t))\right) - \partial_t y_i(t) \times \rho(t,y_i(t)) \leq \frac{\alpha R}{4V} {\left(V-\partial_t y_i(t)\right)}^2, &t \textgreater 0, i \in \mathcal{I}(t),\label{CAV reduction}\\
    &f \left(\rho(t,0)\right) = f_{\text{in}}(t), &t \textgreater 0,\\
    &f\left(\rho(t,L)\right) = f_{\text{out}}(t), &t \textgreater 0,
	\end{align}
\end{subequations}
where $\mathcal{I}(t) = \{0,1,2,\ldots, I\}$ is the set of CAVs at time $t$, and $i\in \mathcal{I}(t)$ is the index of the CAV. The inequality ~(\ref{CAV reduction}) serves as a moving boundary constraint, limiting the flux passing the CAV based on its speed. The parameter $\alpha \in (0,1)$ indicates the reduction in road capacity due to the presence of the CAVs, and we set $\alpha = (W-1)/{W}$, where $W >1$ denotes the number of lanes upstream of the merge point~\citep{Daini2022Centralized}.

To solve the vehicle conservation equation, Eq.~(\ref{PDE-a}) is discretized in both time and space. We set the system update interval as $\Delta t > 0$, and the length of the road cell $j \in \mathcal{J}=\{0,1,\ldots, J\}$ in the coordination zone is $\Delta x >0$. These parameters are chosen to satisfy the Courant-Friedrichs-Lewy (CFL) condition \citep{ZHANG2001337PartB}, $V \Delta t \leq C \Delta x$, where $C \in (0, 1]$ is the Courant number coefficient used to ensure numerical stability. In this study, we set $C=0.9$.
In addition, we define the center position of cell $j$ as $x_j = (j + 1/2) \Delta x$, with the upstream and downstream boundaries given by $x_{j-1/2} = j \Delta x $ and $x_{j+1/2} = (j+1) \Delta x $, respectively. Specifically, $x_{-1⁄2}=0$ and $x_{J+1⁄2}=L$.
At each time step $k \in \{1,2,\ldots\}$, the density $\rho_j(k)$ of each cell $j$, bounded by $\mathcal{D}=[0, R]$, is updated using the supply-demand equation \citep{JIN20121000PartB},
\begin{subequations}
\label{supply-demand}
\begin{align}
&\rho_j({k+1}) = \rho_j(k) - \frac{\Delta t}{\Delta x} \left(F_{j+\frac{1}{2}}(k) -F_{j-\frac{1}{2}}(k) \right),\label{supply-demand-a}\\
&F_{j+\frac{1}{2}}(k) = \min \Big\{ D \left(\rho_j(k) \right), S\left(\rho_{j+1}(k)\right) \Big\},\label{supply-demand-b}\\
&F_{j-\frac{1}{2}}(k) = \min \Big\{ D \left(\rho_{j-1}(k) \right), S\left(\rho_{j}(k)\right) \Big\},\label{supply-demand-c}\\
&D \left(\rho_j(k) \right) = f \left(\min \left\{\rho_j(k), \rho_* \right\} \right),\label{supply-demand-d}\\
&S \left(\rho_j(k) \right) = f \left(\max \left\{\rho_j(k), \rho_* \right\} \right).\label{supply-demand-e}
\end{align}
\end{subequations}
This equation is governed by the upstream flow $F_{j-\frac{1}{2}}(k)$ and the downstream flow $F_{j+\frac{1}{2}}(k)$. And the $F_{-1/2}(k)$ and $F_{J+1/2}(k)$ are calculated using $D\left(\rho_{-1}(k)\right)=f_{\text{in}}(k)$ and $S\left(\rho_{J+1}(k)\right)=f_{\text{out}}(k)$, where $\rho_{-1}(k)= {f_{\text{in}}(k)}/{V}$ and $\rho_{J+1}(k)={f_{\text{out}}(k)}/{V}$. 

\subsection{Problem Definition}
\label{Problem Definition}
Considering the scenario depicted in Fig.~\ref{merging-bottleneck}, CAVs leverage infrastructure communication to acquire real-time road density states and the control intentions of peer agents. Operating upstream of the bottleneck, they predict future traffic evolution and optimize longitudinal control actions to proactively regulate the inflow of following HDVs. This regulation strategy serves to dampen density oscillations and maintain smooth spacing, thereby mitigating congestion triggers at the merge point and maximizing overall road throughput.

However, in mixed autonomy traffic flow, the speed control of CAVs may introduce discontinuities in the surrounding traffic density. To investigate  this problem, at time step $k$, we consider the upstream and downstream densities of the cell $j^{'}$ where CAV $i \in \mathcal{I}(k)$ is located as Riemann-type initial reference. When the vehicle moves at the desired speed $u_i(k)$ and the constraint (\ref{CAV reduction}) is not satisfied (such CAVs are referred to as controlled CAVs) , its upstream and downstream densities are updated to $\hat{\rho}_{j^{'}-\frac{1}{2}}(u_i(k))$ and $\check{\rho}_{j^{'}+\frac{1}{2}}(u_i(k))$ respectively~\citep{Daini2022Centralized}. These updated values simulate the congestion effect induced by the controlled CAV.
The flow between the cell $j^{'}$ and $j^{'}-1$, $j^{'}+1$ can be reconstructed as,
\begin{align}
& F_{j'-\frac{1}{2}}(k) = \min \left \{ D \left(\rho_{j'-1}(k) \right), S\left(\hat{\rho}_{j'-\frac{1}{2}}(u_i(k))\right) \right\},\label{CAV_up}\\
& \Delta t \times F_{j'+\frac{1}{2}}(k) = \min \left \{ \Delta t_{i,j'}(k), \Delta t \right\} f\left( \check{\rho}_{j'+\frac{1}{2}}(u_i(k))\right) \nonumber\\
& \qquad \qquad \qquad \quad + \max \left \{ \Delta t - \Delta t_{i,j'}(k), 0 \right\} f\left( \hat{\rho}_{j'-\frac{1}{2}}(u_i(k))\right),\label{CAV_down}
\end{align}
where $\Delta t_{i,j^{'}}(k)$ represents the time required for CAV $i$ to reach the downstream $x_{j^{'}+1/2}$, 
\begin{align*}
    & \Delta t_{i,j^{'}}(k) = \frac{1-d_{i,j^{'}}(k)}{u_i(k)} \Delta x, \\
    & d_{i,j^{'}}(k) = \frac{\rho_{j^{'}}(k)-\hat{\rho}_{j^{'}-\frac{1}{2}}(u_i(k))}{\check{\rho}_{j^{'}+\frac{1}{2}}(u_i(k))-\hat{\rho}_{j^{'}-\frac{1}{2}}(u_i(k))}.
\end{align*}
This reconstruction process of density and flow directly reflects the constraint (\ref{CAV reduction}), and therefore, it will not be restated explicitly in the subsequent discussion.
When multiple CAVs are in the same cell, they are checked sequentially for compliance with constraint (\ref{CAV reduction}). If there is more than one controlled CAV, the controlled CAV closest to the downstream is used for flow reconstruction, as it has a greater impact on the traffic flow. 

Consequently, for each time step $k$, the manner in which cell densities are updated within the multiagent system depends on the speed command $u(k)$ executed by the CAV, and is formally described as follows,
\begin{align*}
\rho_j(k+1) = \rho_j(k) +& f_1\Big(\rho_{j-1}(k),\rho_{j}(k),\rho_{j+1}(k)\Big), \nonumber \\
&\qquad \text{if } \forall i \in \mathcal{I}(k), \ y_{i}(k) \notin \Omega_j \text{ or } u_i(k) =0 \text{ or } \neg \Gamma_{i,j}(k),  \\[10pt]
\rho_j(k+1) = \rho_j(k) +& f_2\Big(\rho_{j-1}(k),\rho_{j}(k)\Big) + f_3\left(u_i(k)\right), \nonumber \\
&\qquad \text{if } \exists i \in \mathcal{I}(k), \ y_{i}(k) \in \Omega_j, \ u_i(k) \neq 0 \text{ and } \Gamma_{i,j}(k), 
\end{align*}
where $f_1(\cdot)$, $f_2(\cdot)$, $f_3(\cdot)$ are all nonlinear functions computed from the flow at the cell interfaces and the CAV actions in Eq.~(\ref{supply-demand}). $\Omega_j$ denotes the spatial domain of cell $j$, defined as $ \left[j \Delta x, (j+1)\Delta x \right)$.
Condition $\Gamma_{i,j}(k) := \gamma_1\left(\rho_j(k)\right) < {u}_i(k) < \gamma_2\left(\rho_j(k)\right)$ checks if the control action of CAV $i$ is within the range defined by $\gamma_1\left(\rho_{j}(k) \right)$ and $\gamma_2\left(\rho_{j}(k) \right)$, which violates the constraint (\ref{CAV reduction}).  The condition $\neg \Gamma_{i,j}(k) := {u}_i(k) \leq \gamma_1\left(\rho_j(k)\right) \text{ or } {u_i}(k) \geq \gamma_2\left(\rho_j(k)\right)$ means that CAV $i$ acts as a passive vehicle (or tracer) rather than an active actuator. The detailed derivations are provided in~\ref{Appendix A} and~\ref{Appendix B}.

To unify the description of density transitions for different scenarios, we define the vector $\mathbf{O}(k) = {\left[o_0(k), o_1(k), \ldots, o_J(k)\right]}^T \in \mathbb{R}^{J+1}$, where $o_j(k)=1$ indicates that cell $j$ contains a CAV whose action remains within the maximum allowable downstream limit but does not satisfy the maximum flow reduction constraint at time step $k$, and $o_j(k)=0$ otherwise. The control action of CAV $i$ is also transformed into a vector $\mathbf{u}_i(k) = {\left[u_{i,0}(k), u_{i,1}(k), \ldots, u_{i,J}(k)\right]}^T \in \mathbb{R}^{J+1}$, with non-zero input only in the cell occupied by this CAV. The state transition equation of the overall system can be modeled as,
\begin{align}
\label{overall system}
\bm{\rho}(k+1) = \bm{\rho}(k) + \mathbf{{A}}(k) \mathbf{O}(k) + \mathbf{{B}}(k) + \mathbf{{C}}({\mathbf{u}}(k)),
\end{align}
where $\bm{\rho}(k) = {\left[\rho_0(k), \rho_1(k), \ldots, \rho_J(k)\right]}^T \in \mathbb{R}^{J+1}$ represents the density state of the whole road segment, ${\mathbf{u}}(k) = {\left[\mathbf{u}_0(k), \mathbf{u}_1(k), \ldots, \mathbf{u}_I(k)\right]} \in \mathbb{R}^{(J+1) \times (I+1)}$ denotes the aggregate control actions of all CAVs. The matrix $\mathbf{{A}}(k) \in \mathbb{R}^{(J+1) \times (J+1)}$ and vectors $\mathbf{{B}}(k)\in \mathbb{R}^{J+1}$, $\mathbf{{C}}({\mathbf{u}}(k))\in \mathbb{R}^{J+1}$ are calculated as following,
\begin{align*}
    &\mathbf{{A}}(k) = \mathrm{diag}{\Big(\text{A}_{0,0}(k),{\text{A}}_{1,1}(k),\dots,{\text{A}}_{J,J}(k)\Big)}, \\ 
    &\qquad {\text{A}}_{j,j}(k)=f_2\Big(\rho_{j-1}(k),\rho_{j}(k)\Big) - f_1\Big(\rho_{j-1}(k),\rho_{j}(k),\rho_{j+1}(k)\Big),\\
    &\mathbf{{B}}(k) = {\Big[f_1\Big(\rho_{-1}(k),\rho_{0}(k),\rho_{1}(k)\Big),f_1\Big(\rho_{0}(k),\rho_{1}(k),\rho_{2}(k)\Big),\dots,f_1\Big(\rho_{J-1}(k),\rho_{J}(k),\rho_{J+1}(k)\Big)\Big]}^T,\\
    &\mathbf{{C}}({\mathbf{u}}(k)) ={\Big[{\text{C}}_0({\mathbf{u}(k)}),{\text{C}}_1({\mathbf{u}(k)}),\dots,{\text{C}}_J({\mathbf{u}(k)})\Big]}^T, \\
    &\qquad {\text{C}}_j({\mathbf{u}(k)})=\begin{cases}
    f_3\left(u_{i_j^*,j}(k)\right), & o_{j}(k) =1, \\
    0, & o_{j}(k)= 0 ,
     \end{cases}\\
    & \qquad o_{j}(k)= \begin{cases}
    1, & \exists i \in \mathcal{I}(k), \ y_{i}(k) \in \Omega_j, \, u_{i,j}(k) \neq 0   \ \text{and} \ \Gamma_{i,j}(k), \\
    0, & \forall i \in \mathcal{I}(k), \ y_{i}(k) \notin \Omega_j \ \text{or}\ u_{i,j}(k) = 0  \ \text{or}  \ \neg \Gamma_{i,j}(k),
     \end{cases}\\
    & \qquad  i_j^* = \min \left\{ i \in \mathcal{I}(k) \, \middle| \, y_i(k) \in \Omega_j, \, u_{i,j}(k) \neq 0, \, \Gamma_{i,j}(k) \right\},
\end{align*}
where $i_j^*$ denotes the index of the CAV located in cell $j$ that is closest to the downstream boundary and induces a local capacity reduction. By casting the discretized dynamics into this compact matrix form, we consolidate the complex PDE-ODE interactions into a unified state-space representation. This structural transformation is pivotal, as it streamlines the explicit derivation of the prediction model.

The position $y_i \in [0, L]$ of each CAV $i$ is updated as,
\begin{align*}
y_i(k+1) = y_i(k) + \min\Big\{\max\left(\mathbf{u}_i(k)\right), v\left(\rho(k,y_i(k)+)\right)\Big\} \Delta t.
\end{align*}

Based on the aforementioned state transition equations, we formulate a distributed cooperative control problem for each CAV, where the traffic density of road cells is treated as the state variable, the speed of the CAV serves as the control input, and the objective is to minimize the total travel time of all vehicles within the controlled road segment.
Taking the subsystem of CAV $i$ as an example, at sampling time step $k$, both the prediction horizon and control horizon are set to $N>0$, and the perceived traffic state is denoted by $\bm{\rho}_i(k)=\bm{\rho}(k) \in \bm{\mathcal{D}}$, where $\bm{\mathcal{D}} \in \mathbb{R}^{J+1}$ and each component $\rho_{i,j}(k)= \rho_{j}(k) \in {\mathcal{D}}$, $j \in \mathcal{J}$. The objective function, denoted by, 
\begin{align*}
    J_i\left(\bm{\xi}_i(k),\bm{\mu}_i(k)\right)=\sum_{t=0}^{N-1}\sum_{j=0}^{J}\rho_{i,j}(k+t) \Delta t \Delta x,
\end{align*}
characterizes the total time spent by all vehicles through the evolution of traffic density across road cells over time. Crucially, minimizing this density-based cost inherently mitigates the formation of high-density clusters upstream of the bottleneck. This effectively reduces the turbulence and capacity drop associated with mandatory lane changes and creates favorable gaps for HDVs to merge.
As a system-level cost, it enables each agent to independently optimize its local control actions while collectively contributing to the maximization of overall traffic efficiency. The $\bm{\mu}_i(k) =\left[\mathbf{u}_i(k+0),\mathbf{u}_i(k+1), \ldots, \mathbf{u}_i(k+N-1)\right] \in \mathbb{R}^{(J+1) \times N}$ represents the predicted inputs applied by CAV $i$ from time step $k\Delta t$ to $(k+N-1)\Delta t$, and $\bm{\xi}_i(k) =\left[\bm{\rho}_i(k+1),\bm{\rho}_i(k+2), \ldots, \bm{\rho}_i(k+N)\right] \in \mathbb{R}^{(J+1) \times N}$ denotes the corresponding predicted evolution of traffic states. The prediction model employed during optimization is expressed as,
\begin{align*}
\bm{\rho}_i(k+1) = \bm{\rho}_i(k) + \mathbf{{A}}_i(k) \mathbf{O}_i(k) + \mathbf{{B}}_i(k) + \mathbf{{C}}_{i}(\mathbf{u}(k)).
\end{align*}
Utilizing matrix $\mathbf{{C}}_{i}(\mathbf{u}(k))$, CAV $i$ can explicitly account for the control actions of other agents $j$, thereby ensuring coordination among control actions. 
It is worth noting that while the prediction model established here utilizes deterministic state evolution for computational tractability, the actual traffic environment involves stochastic behaviors. The robustness against such uncertainties is explicitly mitigated by the receding horizon mechanism of the proposed control framework, as detailed in Section~\ref{MPC framework}. 
\section{Optimal Policy with Multiagent Truncated Rollout}
\label{sec_analysis}
In this section, we introduce the optimization framework of the multiagent truncated rollout method (Section~\ref{Framework_Truncated_Rollout}) and present its enhancement in cost performance (Section~\ref{Section_Cost_Improvement}). Additionally, we conduct a stability analysis of the proposed approach (Section~\ref{Stability analysis}).

\subsection{Framework for Multiagent Rollout}
\label{Framework_Truncated_Rollout}
As a special case of the rollout, model predictive control (MPC) enables online prediction of future states and optimization of vehicle actions \citep{ZHANG2023199partb}. Leveraging this property, in this section, we propose a multiagent truncated rollout approach based on the distributed model predictive control (DMPC) framework. 
By integrating an agent-by-agent sequential optimization mechanism with the truncated rollout approach that shortens the optimization horizon, the proposed method enhances optimization in both spatial and temporal dimensions while alleviating computational burden, thereby achieving more efficient and effective control of mixed traffic flow.

\subsubsection{Multiagent Model Predictive Control}
\label{MPC framework}
In the multiagent truncated rollout method, each CAV is equipped with its own controller. These controllers receive driving information from all other CAVs, predict the evolution of traffic states, and independently solve an optimization problem to obtain an optimal control sequence. Specifically, for each CAV $i \in \mathcal{I}(k)$, the travel cost optimization problem can be formulated as,
\begin{subequations}
\label{control}
\begin{align}
&\bm{\mu}_i^{*}(k) = \arg\min_{\bm{\mu}_i(k)} J_i\left(\bm{\xi}_i(k),\bm{\mu}_i(k)\right), \label{control-a}\\
\text{s.t.} \quad &\bm{\rho}_i(k+t+1)=\bm{\rho}_i(k+t)+ \mathbf{{A}}_i(k+t)\mathbf{O}_i(k+t)+\mathbf{{B}}_i(k+t)+\mathbf{{C}}_i({\mathbf{u}(k+t)}),\label{control-b}\\
& y_{i'}(k+t+1) = y_{i'}(k+t) + \min\Big\{\max\left(\mathbf{u}_{i'}(k+t)\right), v\left(\rho(k+t,y_{i'}(k+t)+)\right)\Big\} \Delta t,  \nonumber \\
& \qquad {i'} \in \mathcal{M}_i(k),\label{control-c}\\
&o_{i,j}(k+t)  = \nonumber \\ 
&\qquad \quad \begin{cases}
    1, & \exists \ i^{'} \in \mathcal{M}_i(k), \ y_{i^{'}}(k+t) \in \Omega_j, \, u_{i^{'},j}(k+t) \neq 0   \ \text{and} \ \Gamma_{i^{'},j}(k+t), \\
    0, & \forall \ i^{'} \in \mathcal{M}_i(k), \ y_{i^{'}}(k+t) \notin \Omega_j \ \text{or}\ u_{i^{'},j}(k+t) = 0  \ \text{or}  \ \neg \Gamma_{i^{'},j}(k+t),
     \end{cases} \label{control-d}  \\
&V_i\left(\bm{\xi}_i(k),\bm{\mu}_i(k)\right) \leq \eta_i\left(\bm{\xi}_i(k),\lambda_i\right) \label{control-g}, \\
&\bm{\rho}_i(k) = \bm{\rho}(k), \quad \bm{\rho}_{i}(k+t) \in \bm{\mathcal{D}}, \quad \bm{\rho}_{i}(k+N) \in \bm{\mathcal{D}}_{i,T},  \quad  \max\left(\mathbf{u}_{i}(k+t)\right) \in U, \label{control-e}\\
&t \in \{0,1,\ldots, N-1\}, \quad  j \in \mathcal{J},\label{control-f}
\end{align}
\end{subequations}
where $\bm{\mu}_i^{*}(k)$ represents the optimal cost control sequence at time $k$, while $\bm{\xi}_i^{*}(k)$ denotes the corresponding prediction state sequence. $\mathcal{M}_i(k) \in \mathcal{I}(k)$ represents the set of CAVs for which control sequences have already been computed when solving subsystem $i$ at time step $k$. Since the CAVs execute an agent-by-agent sequential decision-making approach, $\mathcal{M}_i(k)$ remains unchanged within the prediction horizon. The position update of CAV $i'\in \mathcal{M}_i(k)$ follows the constraint~(\ref{control-c}). And CAV $i^{''} \in \mathcal{I}(k)-\mathcal{M}_i(k)$ that have not yet obtained their control sequences will temporarily follow the same car-following model as HDVs. $\mathbf{O}_i(k+t)={\left[o_{i,0}(k+t),o_{i,1}(k+t), \ldots,o_{i,J}(k+t)\right]}^T$ indicates whether each cell $j$ contains any CAV bottleneck.
The initial condition is $\bm{\rho}_i(k) = \bm{\rho}(k)$, and the terminal constraint is $\bm{\rho}_{i}(k+N) \in \bm{\mathcal{D}}_{i,T}$, where $\bm{\mathcal{D}}_{i,T} \in \mathbb{R}^{J+1}$ is the terminal constraint set for CAV $i$, and its solution method is detailed in~\citep{Ferramosca2009nonlinearMPC}.
The sets of states and actions contain certain equilibrium points $\left(\bm{\rho}_i^e,\mathbf{u}_i^e\right)$. This paper assumes that the optimal cost equilibrium point is the origin. The function $\eta_i\left(\bm{\xi}_i(k),\lambda_i\right)$ in constraint~(\ref{control-g}) is the contraction function that needs to be designed. 
The stability objective $V_i\left(\bm{\xi}_i(k), \bm{\mu}_i(k)\right)$ is a classical MPC tracking cost function, designed to ensure that the system gradually converges to the desired state without divergence,
\begin{align}
    & V_i\left(\bm{\xi}_i(k),\bm{\mu}_i(k)\right) = \sum_{t=0}^{N-1}L_s\left(\bm{\rho}_i(k+t),\mathbf{u}_i(k+t)\right) +E\left(\bm{\rho}_i(k+N)\right), \label{valuefunction} \\
    & L_s\left(\bm{\rho}_i(k+t),\mathbf{u}_i(k+t)\right)={\bm{\rho}_i(k+t)}^T \mathbf{Q}_i \bm{\rho}_i(k+t) + {\mathbf{u}_i(k+t)}^T \mathbf{R}_i \mathbf{u}_i(k+t), \\
    & E\left(\bm{\rho}_i(k+N)\right)={\bm{\rho}_i(k+N)}^T \mathbf{H}_i \bm{\rho}_i(k+N), \label{E_rho_k}
\end{align}
where the matrices $\mathbf{H}_i \in \mathbb{R}^{(J+1)\times (J+1)}$, $\mathbf{Q}_i \in \mathbb{R}^{(J+1)\times (J+1)}$, and $\mathbf{R}_i \in \mathbb{R}^{(J+1)\times (J+1)}$ are all positive definite and symmetric. $\mathbf{Q}_i$ and $\mathbf{R}_i$ are the given state and control weighting matrices, respectively. $\mathbf{H}_i$ is the terminal state weighting matrix, which satisfies the Lyapunov equation \citep{Magni2006Regional},
\begin{align}
\label{Lyapunov equation}
    {\mathbf{Z}_i}^T \mathbf{H}_i \mathbf{Z}_i - \mathbf{H}_i = - (\mathbf{Q}_i^* +\Delta \mathbf{Q}_i), \quad \mathbf{Q}_i^* = \mathbf{Q}_i + {\mathbf{K}_i}^T \mathbf{R}_i \mathbf{K}_i,
\end{align}
$\Delta \mathbf{Q}_i \in \mathbb{R}^{(J+1)\times (J+1)}$ is a positive definite matrix, $\mathbf{Z}_i \in \mathbb{R}^{(J+1)\times (J+1)}$ and $\mathbf{K}_i \in \mathbb{R}^{(J+1)\times (J+1)}$ are defined as detailed in Section \ref{Stability analysis}.

To construct the contraction function $\eta_i(\bm{\xi}_i(k),\lambda_i)$, we define the optimization problem for the stability function as,
    \begin{align}
	&\bm{\mu}_i^{v}(k) = \arg\min_{\bm{\mu}_i(k)} V_i\left(\bm{\xi}_i(k),\bm{\mu}_i(k)\right), \nonumber \\
	\text{s.t.} \quad & \text{constraints~(\ref{control-b})-(\ref{control-d})}, ~(\ref{control-e})-(\ref{control-f}),\label{value function optimal}
	\end{align}
where $\mu_i^{v}(k)$ represents the optimal stability solution at time $k$. 
Based on $\mu_i^{v}(k)$, the optimal control sequence $\bm{\mu}_i^{*}(k-1)$ and the corresponding predicted state sequence $\bm{\xi}_i^{*}(k-1)$ obtained at time step $k-1$ from the control problem in Eq.~(\ref{control}), we define $\eta_i\left(\bm{\xi}_i(k), \lambda_i\right)$ as
\begin{align}
\eta_i\left(\bm{\xi}_i(k),\lambda_i\right)=V_i\left(\bm{\xi}_i^{v}(k),\bm{\mu}_i^{v}(k)\right) + \lambda_i\left[V_i\left(\bm{\xi}_i^{*}(k-1),\bm{\mu}_i^{*}(k-1)\right)-V_i\left(\bm{\xi}_i^{v}(k),\bm{\mu}_i^{v}(k)\right)\right],
\end{align}
and $\lambda_i \geq 0$. $\bm{\xi}_i^{v}(k)$ represents the predicted state sequence corresponding to the predictive sequence $\bm{\mu}_i^{v}(k)$. 

At each sampling time $k$, CAV $i \in \mathcal{I}(k)$ first minimizes the stability objective function to construct the contraction constraint, and then explores the optimal cost control sequence. If the control problem in Eq.~(\ref{control}) for subsystem $i$ is feasible, the first control input $\mathbf{u}_i^{*}(k+0)=\mathbf{u}_i^{*}(k)$ of the final solution $\bm{\mu}_i^{*}(k)$ is defined as the distributed coupled model predictive control law, $\mathbf{u}_i(k)=\mathbf{u}_i^{*}(k)$.
The closed-loop system corresponding to CAV $i$  is,
\begin{align}
\label{closed-loop}
\bm{\rho}_i(k+1)=\bm{\rho}_i(k)+ \mathbf{{A}}_i(k)\mathbf{O}_i(k)+\mathbf{{B}}_i(k)+\mathbf{{C}}_i({\mathbf{u}}(k)) = \mathbf{g}_i\left(\bm{\rho}_i(k), \mathbf{u}_i(k), \mathbf{e}_i(k)\right).
\end{align}
Both $\mathbf{O}_i(k)$ and $\mathbf{{C}}_i({\mathbf{u}}(k)) $ include the position of CAV $i$ and the travel information of all other CAVs. For convenience, we denote this additional information by $\mathbf{e}_i(k)$.

Although the density prediction model employed in the optimization framework does not explicitly account for the stochasticity of HDVs, the optimization problem is re-initialized at each time step $k$ based on the actual measured state $\bm{\rho}(k)$. This feedback loop enables the controller to continuously compensate for prediction errors arising from HDV stochasticity or model mismatches, thereby ensuring the robustness of the control strategy in mixed traffic environments. 

\subsubsection{Sequential Optimization with Truncated Rollout}
In multiagent systems, conventional DMPC, while computationally efficient, typically yields suboptimal solutions due to the lack of global coordination. To bridge this gap, we propose an agent-by-agent iterative optimization approach that maintains the distributed nature of DMPC while explicitly enhancing coordination among CAVs through real-time policy sharing.

\begin{figure}[htbp]
  \centering
  \subfigure[Agent-by-agent optimality]
  {\label{Agent-by-agent}
  \includegraphics[width=0.9\textwidth,trim=0 0 0 0, clip]{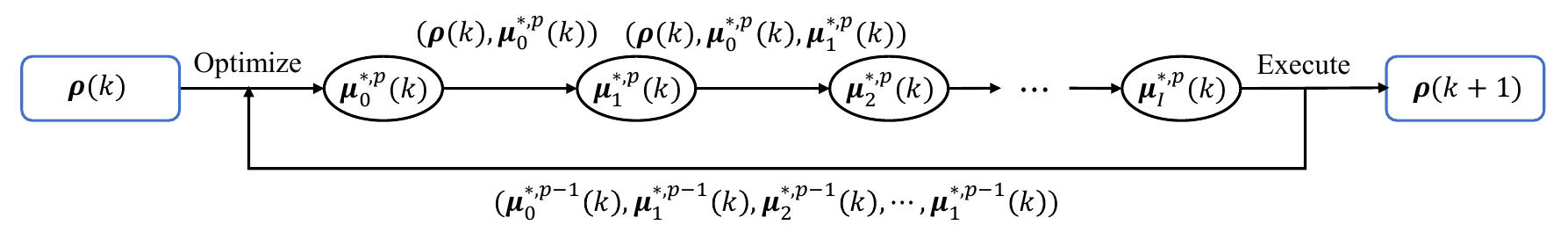}}
  \subfigure[Truncated rollout strategy in each iteration]
  {\label{Truncated-rollout}
  \includegraphics[width=0.9\textwidth,trim=0 0 0 0, clip]{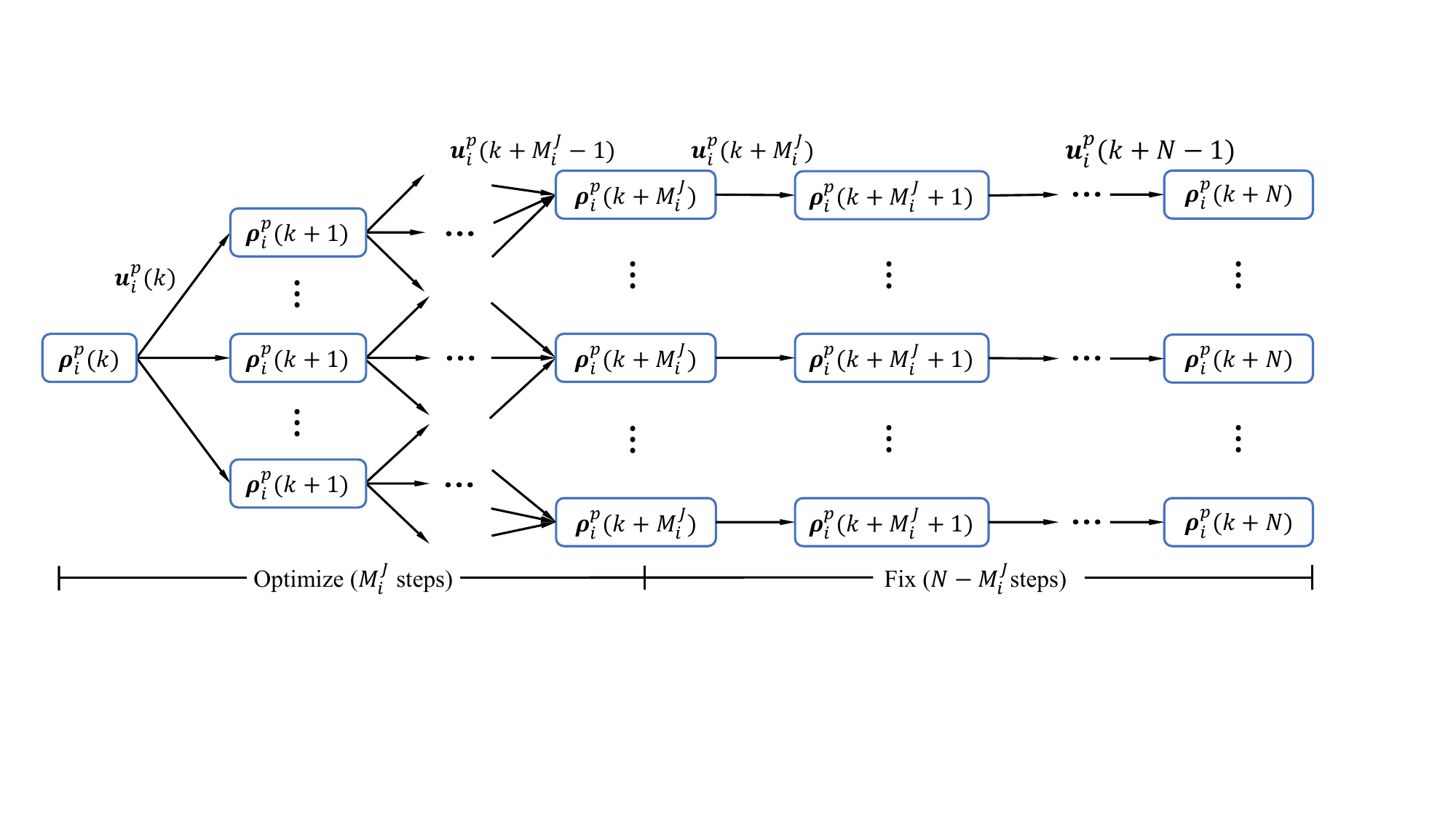}} 
  \caption{The framework for multiagent truncated rollout approach.} \label{framework}
\end{figure}

As illustrated in Fig.~\ref{Agent-by-agent}, at each time step $k$, the CAVs within the set $\mathcal{I}(k)=\{0,1,\ldots,I\}$ sequentially compute their control sequences according to a predetermined order. Iterative refinement is performed after all CAVs complete one round of policy computation to progressively converge towards the optimal control solution. In this section, we assume CAVs execute their optimization in descending index order (i.e., from highest to lowest index). Specifically, during the $p$-th iteration, each CAV $i \in \mathcal{I}(k)$ receives the latest control sequences $\bm{\mu}_{0:i-1}^{*,p}(k)$ from CAV $0$ to $i-1$, and the historical control strategies $\bm{\mu}_{i+1:I}^{*,p-1}(k)$ from CAV $i+1$ to $I$. Leveraging the density state transition equation to predict future traffic states, CAV $i$ then solves the optimization problem in Eq.~(\ref{control}) to obtain its optimal control strategy $\bm{\mu}_{i}^{*,p}(k)$. Subsequently, CAV $i$ propagates this updated policy to its succeeding CAVs. This process establishes a sequential cooperative decision chain, ensuring real-time responsiveness and collaborative effectiveness.

To enhance optimization performance while improving computational efficiency, we propose a truncated rollout optimization strategy (see Fig.~\ref{Truncated-rollout}), which reduces the dimensionality of the optimization variables by truncating the prediction horizon during each iteration \citep{Bertsekas2020ConstrainedMR}. At time step $k$, the CAVs first execute an $N$-step prediction to obtain an initial feasible solution $\bm{\mu}_{i}^{*,0} (k)=\left[\mathbf{u}_{i}^{*,0}(k+0),\mathbf{u}_{i}^{*,0}(k+1), \ldots, \mathbf{u}_{i}^{*,0}(k+N-1)\right]$, ensuring the system can make reasonable adjustments in the early stages. During the iteration $p \geq 1$, the prevailing control sequence $\bm{\mu}_{i}^{*,p-1}(k)$ of CAV $i$ is divided into two parts: the first $M_i^J$ steps (where $M_i^J \leq N$) for optimization, denoted as $\bm{\mu}_{i,f}^{*,p-1}(k)$; and the remaining $N-M_i^J$ steps for fixing, denoted as $\bm{\mu}_{i,r}^{*,p-1}(k)$,
\begin{align*}
    &\bm{\mu}_{i,f}^{*, p-1}(k) = \left[\mathbf{u}_{i}^{*, p-1}(k+0),\mathbf{u}_{i}^{*, p-1}(k+1),\ldots,\mathbf{u}_{i}^{*, p-1}(k+M_i^J-1)\right],\\
    &\bm{\mu}_{i,r}^{*,p-1}(k) = \left[\mathbf{u}_{i}^{*,p-1}(k+M_i^J),\mathbf{u}_{i}^{*, p-1}(k+M_i^J+1),\ldots,\mathbf{u}_{i}^{*, p-1}(k+N-1)\right].
\end{align*}
The controller of CAV $i$ optimizes the first $M_i^J$ control steps based on the current environmental state and $\bm{\mu}_{i,r}^{*,p-1}(k)$, while $\bm{\mu}_{i,r}^{*,p-1}(k)$ remains unchanged throughout this process.  This maintains the control sequence length at $N$ and reduces computational complexity.
The value of $M_i^J$ depends on the current system objective function $J_i\left(\bm{\xi}_i^{*}(k),\bm{\mu}_i^{*}(k)\right)$,
\begin{align}
\label{M_value_J}
M_i^J = M_{\min} + (N - M_{\min}) \frac{\max \left\{J_i^{LB}(k), \min \left\{J_i \left(\bm{\xi}_i^{*}(k), \bm{\mu}_i^{*}(k)\right), J_i^{UB}(k) \right\} \right\} - J_i^{LB}(k)}{J_i^{UB}(k) - J_i^{LB}(k)},
\end{align}
where $M_{\min}$ represents the minimum truncation horizon, $J_i^{UB}(k)$ and $J_i^{LB}(k)$ denote the upper and lower bounds of $J_i\left(\bm{\xi}_i(k),\bm{\mu}_i(k)\right)$, respectively. When the current objective value approaches $J_i^{UB}(k)$, it reveals a significant suboptimality gap between the incumbent cost control policy $\bm{\mu}_i^{*}(k)$ and the theoretical Pareto-optimal solution. In this case, increasing the truncation length $M_i^J$ enhances the global search capability of the optimization problem. Conversely, proximity of $J_i\left(\bm{\xi}_i^{*}(k),\bm{\mu}_i^{*}(k)\right)$ to its lower bound $J_i^{LB}(k)$ implies that the current policy $\bm{\mu}_i^{*}(k)$ resides within a neighborhood of the local optimum, where the benefit of extended-horizon optimization diminishes. In this situation, adaptively decreasing $M_i^J$ can reduce the computational cost while further approximating the globally optimal solution.
The stability function $V_i\left(\bm{\xi}_i(k),\bm{\mu}_i(k)\right)$ is solved iteratively using the same truncation strategy, and
\begin{align}
\label{M_value_V}
M_i^V = M_{\min} + (N - M_{\min}) \frac{\max \left\{V_i^{LB}(k), \min \left\{V_i \left(\bm{\xi}_i^{v}(k), \bm{\mu}_i^{v}(k)\right), V_i^{k,UB} \right\} \right\} - V_i^{LB}(k)}{V_i^{UB}(k)- V_i^{LB}(k)}.
\end{align}

Based on the sequential agent-by-agent optimization property, the state prediction of the last CAV $I$ in Eq.~(\ref{control-b}) incorporates the control information propagated from all CAVs in the current iteration. Consequently, the optimization process terminates when the change in the objective value of CAV $I$ between two consecutive iterations falls below a predefined threshold  $\epsilon$,
\begin{align*}
    \left| J_I\left(\bm{\xi}_{I}^{*,p}(k),\bm{\mu}_{I}^{*,p}(k)\right) - J_I\left(\bm{\xi}_{I}^{*,p-1}(k),\bm{\mu}_{I}^{*,p-1}(k)\right) \right| \leq \epsilon.
\end{align*}

\subsection{Cost Improvement Properties}
\label{Section_Cost_Improvement}
In this section, the performance of cost improvement in the multiagent truncated rollout is analyzed through the following propositions. 
\begin{prp}
\label{agent-by-agent}
The agent-by-agent sequential decision-making approach, which couples the control actions of other CAVs, is no worse than the parallel decision-making approach.
\end{prp}
\begin{proof}
Assume that at time step $k$, the decisions made by all CAVs in the $p$-th iteration under the parallel decision-making method are $\bm{\mu}_0^{'*,p}(k),\bm{\mu}_1^{'*,p}(k),\ldots, \bm{\mu}_I^{'*,p}(k)$, while those made by the agent-by-agent sequential decision-making method are $\bm{\mu}_0^{*,p}(k),\bm{\mu}_1^{*,p}(k),\ldots, \bm{\mu}_I^{*,p}(k)$.  All actions in the control sequences $\bm{\mu}_i^{'*,p}(k)$ and $\bm{\mu}_i^{*,p}(k)$ of any CAV $i \in \mathcal{I}(k)$ are derived from the same control action set $U$. The corresponding system objective values are $J\left(\bm{\xi}^{'*,p}(k),\bm{\mu}_0^{'*,p}(k),\bm{\mu}_1^{'*,p}(k),\ldots, \bm{\mu}_I^{'*,p}(k)\right)$ and $J\left(\bm{\xi}^{*,p}(k), \bm{\mu}_0^{*,p}(k),\bm{\mu}_1^{*,p}(k),\ldots, \bm{\mu}_I^{*,p}(k)\right)$. 
Furthermore, we assume that both methods are based on the optimal solution $\bm{\mu}_0^{*,p-1}(k),\bm{\mu}_1^{*,p-1}(k)
$, $\ldots, \bm{\mu}_I^{*,p-1}(k)$ found in the $(p-1)$-th iteration.

According to the optimization principles of both methods, the cost improvement for the first CAV relies on the results of the previous iteration,
\begin{align*}
    \bm{\mu}_0^{*,p}(k) = \bm{\mu}_0^{'*,p}(k)= \arg \min_{\bm{\mu}_0(k)} J\left(\bm{\xi}_0(k), \bm{\mu}_0(k), \bm{\mu}_1^{*,p-1}(k), \dots, \bm{\mu}_I^{*,p-1} (k)\right).
\end{align*}
Thus, the system objective values satisfy,
\begin{align*}
    J\left(\bm{\xi}_0(k), \bm{\mu}_0^{'*,p}(k),\bm{\mu}_1^{*,p-1}(k),\ldots, \bm{\mu}_I^{*,p-1}(k)\right) = J\left(\bm{\xi}_0(k), \bm{\mu}_0^{*,p}(k),\bm{\mu}_1^{*,p-1}(k),\ldots, \bm{\mu}_I^{*,p-1}(k)\right).
\end{align*}

For the optimization of strategies for subsequent CAVs, during parallel decision-making, they still execute the control sequence obtained from the $(p-1)$-th iteration. In contrast, within the sequential decision-making approach, the latest policies of the preceding vehicles in the current iteration is used to replace the results from the previous iteration.
This means that during agent-by-agent decision-making, CAVs not only consider the environmental state but also coordinate in real time with the actions of preceding CAVs, ensuring a more optimal overall system objective. 
As an example, for the second CAV, the sequential method optimizes by considering the latest strategy $\bm{\mu}_0^{*,p}(k)$ of the first CAV, whereas the parallel method only uses the $\bm{\mu}_0^{*,p-1}(k)$,
\begin{align*}
  &  \min_{\bm{\mu}_1(k)}J\left(\bm{\xi}_1(k),\bm{\mu}_0^{*,p-1}(k),\bm{\mu}_1(k),\bm{\mu}_2^{*,p-1}(k), \ldots,\bm{\mu}_I^{*,p-1}(k) \right) \\
    &= J\left(\bm{\xi}_1(k),\bm{\mu}_0^{*,p-1}(k),\bm{\mu}_1^{'*,p}(k),\bm{\mu}_2^{*,p-1}(k), \ldots,\bm{\mu}_I^{*,p-1}(k) \right) \\
    & \geq J\left(\bm{\xi}_1(k),\bm{\mu}_0^{'*,p}(k),\bm{\mu}_1^{'*,p}(k),\bm{\mu}_2^{*,p-1}(k), \ldots,\bm{\mu}_I^{*,p-1}(k)\right)\\
    & \geq \min_{\bm{\mu}_1(k)}J\left(\bm{\xi}_1(k),\bm{\mu}_0^{*,p}(k),\bm{\mu}_1(k),\bm{\mu}_2^{*,p-1}(k), \ldots,\bm{\mu}_I^{*,p-1}(k) \right) \\
    & = J\left(\bm{\xi}_1(k),\bm{\mu}_0^{*,p}(k),\bm{\mu}_1^{*,p}(k), \bm{\mu}_2^{*,p-1}(k),\ldots,\bm{\mu}_I^{*,p-1}(k) \right).
\end{align*}

Therefore, extending this reasoning to the entire system, we can obtain,
\begin{align*}
   & \min_{\bm{\mu}_I(k)}J\left(\bm{\xi}_I(k),\bm{\mu}_0^{*,p-1}(k),\bm{\mu}_1^{*,p-1}(k), \ldots,\bm{\mu}_{I-1}^{*,p-1}(k),\bm{\mu}_I(k) \right) \\
    &= J\left(\bm{\xi}_I(k),\bm{\mu}_0^{*,p-1}(k),\bm{\mu}_1^{*,p-1}(k), \ldots,\bm{\mu}_{I-1}^{*,p-1}(k),\bm{\mu}_I^{'*,p}(k) \right) \\
    & \geq J\left(\bm{\xi}_I^{'*,p}(k),\bm{\mu}_0^{'*,p}(k),\bm{\mu}_1^{'*,p}(k), \ldots,\bm{\mu}_{I-1}^{'*,p}(k),\bm{\mu}_I^{'*,p}(k) \right)\\
    & \geq \min_{\bm{\mu}_I(k)}J\left(\bm{\xi}_I(k),\bm{\mu}_0^{*,p}(k),\bm{\mu}_1^{*,p}(k), \ldots,\bm{\mu}_{I-1}^{*,p}(k),\bm{\mu}_I(k) \right) \\
    & = J\left(\bm{\xi}_I^{*,p}(k),\bm{\mu}_0^{*,p}(k),\bm{\mu}_1^{*,p}(k),\ldots, \bm{\mu}_{I-1}^{*,p}(k),\bm{\mu}_I^{*,p}(k) \right),
\end{align*}
which shows that the system performance resulting from sequential decision-making will not be worse than that from parallel decision-making.
\end{proof}

\begin{prp}
\label{truncated-rollout}
The truncated rollout optimization method can significantly enhance the asymptotic performance of the control strategy and ensure a strict monotonic decrease in the system objective function.
\end{prp}
\begin{proof}
Assume that at any time step $k$, the initial objective function of CAV $i$ is given by $J_i(\bm{\xi}_i^{*,0}(k),\bm{\mu}_{i}^{*,0}(k))$. To minimize $J_i(\bm{\xi}_i(k),\bm{\mu}_{i}(k))$, the first $M_i^J$ steps of the control sequence are optimized in the action set $U$ through truncated rollout approach,
\begin{align*}
    \bm{\mu}_{i,f}^{*,1}(k) = \arg\min_{\bm{\mu}_{i,f}(k)} J_i\left(\bm{\xi}_i(k),\bm{\mu}_{i,f}(k), \bm{\mu}_{i,r}^{*,0}(k)\right).
\end{align*}
The later part of the control sequence $\bm{\mu}_{i,r}^{*,0}(k)$ remains unchanged, and the new optimized control sequence is formed,
\begin{align*}
    \bm{\mu}_{i}^{*,1}(k) = \left[\bm{\mu}_{i,f}^{*,1}(k), \bm{\mu}_{i,r}^{*,0}(k)\right].
\end{align*}
This process ensures that, at any prediction step $k+t$, the density
$\bm{\rho}_{i}^{1}(k+t) \in \bm{\mathcal{D}}$, where $t \in \{0,1,\ldots,N\}$. The $\max\left(\mathbf{u}_i^{1}(k+t)\right) \in U$ satisfies the sequential improvement assumption of the rollout method \citep{Bertsekas1997RolloutAF}, guaranteeing monotonic improvement of the system objective,
\begin{align*}
    J_i\left(\bm{\xi}_i^{*,0}(k),\bm{\mu}_{i}^{*,0}(k)\right) \geq J_i\left(\bm{\xi}_i^{*,1}(k),\bm{\mu}_{i}^{*,1}(k)\right).
\end{align*}

As the iterations progress, the control sequence of CAV $i$ gradually approaches an approximate optimal solution,
\begin{align*}
    J_i\left(\bm{\xi}_i^{*,0}(k),\bm{\mu}_{i}^{*,0}(k)\right) \geq J_i\left(\bm{\xi}_i^{*,1}(k),\bm{\mu}_{i}^{*,1}(k)\right) \geq \ldots \geq J_i\left(\bm{\xi}_i^{*,p}(k),\bm{\mu}_{i}^{*,p}(k)\right),
\end{align*}
i.e., the truncated rollout strategy can effectively improve the control sequence through iterations, achieving superior control performance.
\end{proof}

\subsection{Stability Analysis}
\label{Stability analysis}
In this section, we will prove the stability of the established nonlinear system in Eq.~(\ref{closed-loop}). It should be noted that, although the truncated rollout optimizes only the first $M$ steps of the iteration, the constraints of the optimization problem ensure that the complete control sequence, consisting of the optimized and fixed parts, does not affect the stability of the system. We provide several fundamental definitions and assumptions \citep{Luo2022Multiobjective},

\begin{dfn}
\label{definition 1}
Given a discrete-time system $\bm{\rho}_i(k+1)=\mathbf{g}_i\left(\bm{\rho}_i(k), \mathbf{u}_i(k), \mathbf{e}_i(k)\right)$, a set $\bm{\mathcal{D}} \subseteq \mathbb{R}^{J+1}$ is called an invariant set if for any state $\bm{\rho}_i(k) \in \bm{\mathcal{D}}$ and all admissible control inputs $\max\left(\mathbf{u}_i(k)\right) \in U$, the subsequent state satisfies $\bm{\rho}_i(k+1) \in \bm{\mathcal{D}}$.
\end{dfn}

\begin{dfn}
\label{definition 2}
For the system in Definition~\ref{definition 1} with invariant set $\bm{\mathcal{D}}$, the system is input-to-state stability (ISS) within $\bm{\mathcal{D}}$ if for any initial condition $\bm{\rho}_i(k) \in \bm{\mathcal{D}}$, there exists a KL-class function $\bm{\beta}(\cdot):\mathbb{R}_{\geq 0} \times \mathbb{N} \rightarrow \mathbb{R}_{\geq 0}$ such that the state trajectory satisfies \citep{Magni2006Regional},
\begin{align*}
   \left \|\bm{\rho}_i(k+t)\right\|_{\bm{\mathcal{D}}} \leq \bm{\beta}\left(\left\|\bm{\rho}_i(k)\right\|, t\right),  \forall t \in \{0,1,2,\ldots, N-1\},
\end{align*}
where $\bm{\rho}_i(k+t)$ denotes the state response at time $k+t$ with initial condition $\bm{\rho}_i(k)$, $\left \|\cdot\right\|$ is the induced norm of the set $\bm{\mathcal{D}}$.
\end{dfn}

\begin{lmm}
\label{Lemma 1}
For the system $\bm{\rho}_i(k+1)=\mathbf{g}_i\left(\bm{\rho}_i(k), \mathbf{u}_i(k), \mathbf{e}_i(k)\right)$ and invariant set $\bm{\mathcal{D}} \subseteq \mathbb{R}^{J+1}$, if there exists a function $ V_i\left(\bm{\xi}_i(k), \bm{\mu}_i(k)\right) : \mathbb{R}^{J+1} \times \mathbb{R}^{J+1} \to \mathbb{R}_{\geq 0}$ such that for all $\bm{\rho}_i(k) \in \bm{\mathcal{D}}$ the following conditions hold,
\begin{align*}
    &\alpha_1\left(\left\| \bm{\rho}_i(k) \right\|\right) \leq  V_i\left(\bm{\xi}_i(k), \bm{\mu}_i(k)\right)
     \leq \alpha_2\left(\left\| \bm{\rho}_i(k) \right\|\right), \\
    & V_i\left(\bm{\xi}_i(k+1), \bm{\mu}_i(k+1)\right)- V_i\left(\bm{\xi}_i(k), \bm{\mu}_i(k)\right) \leq -\alpha_3\left(\left\| \bm{\rho}_i(k) \right\|\right),
\end{align*}
where $\alpha_1(\cdot)$, $\alpha_2(\cdot)$, and $\alpha_3(\cdot)$ are $K_\infty$ class functions, then $V_i\left(\bm{\xi}_i(k), \bm{\mu}_i(k)\right)$ is referred to as an ISS-Lyapunov function for the system, and the system is ISS within the invariant set $\bm{\mathcal{D}}$.
\end{lmm}

\begin{dfn}
\label{definition 3}
Consider the closed-loop system  $\bm{\rho}_i(k+1)=\mathbf{g}_i\left(\bm{\rho}_i(k), \mathbf{u}_i(k), \mathbf{e}_i(k)\right)$, and let $\bm{\mathcal{D}} \subseteq \mathbb{R}^{J+1}$ be an invariant set. A state $\bm{\rho}_i(k) \in \bm{\mathcal{D}}$ is called a feasible initial state if there exists a feasible predictive control sequence $\bm{\mu}_i(k)=\left[\mathbf{u}_i(k),\mathbf{u}_i(k+1),\ldots,\mathbf{u}_i(k+N-1)\right]$ such that the closed-loop trajectory satisfies $\bm{\rho}_i(k+t) \in \bm{\mathcal{D}}$ for all $t \in \{0,1,\ldots,N-1\}$. The set of all such feasible initial states is denoted by $\bm{\mathcal{D}}_N$, which satisfies $\bm{\mathcal{D}}_{i,T}\subseteq \bm{\mathcal{D}}_N \subseteq \bm{\mathcal{D}}$, where $\bm{\mathcal{D}}_{i,T}$ denotes the terminal constraint set.
\end{dfn}

\begin{asm}
\label{Assumption 1}
The environmental state $\bm{\rho}(k)$ is globally observable, and all subsystems share identical environmental states,
\begin{align*}
\bm{\rho}_0(k)=\bm{\rho}_1(k)=\ldots=\bm{\rho}_I(k)=\bm{\rho}(k).
\end{align*}
\end{asm}

\begin{asm}
\label{Assumption 3}
During the prediction horizon $[k,k+N]$, the number of CAVs in the system remains constant, that is, no new CAVs are allowed to enter. For any CAV that exits the system at time step $k+t$ within the prediction horizon, its control inputs are assumed to be zero for the remaining time steps $[k+t+1,k+N]$.
\end{asm}

\begin{asm}
\label{Assumption 2}
There exists a local linear feedback control law $\mathbf{u}_i(k)=\mathbf{K}_i \bm{\rho}_i(k)$, where $\mathbf{K}_i \in \mathbb{R}^{(J+1)\times (J+1)}$, such that $\max\left(\mathbf{u}_i(k)\right) \in U$ holds for all $\bm{\rho}_i(k) \in \bm{\mathcal{D}_{i,T}}$. The design of this control law is based on the linearization of the system in Eq.~(\ref{closed-loop}) around the origin (an asymptotically stable equilibrium point), given by,
\begin{align}
\label{linearization}
    \bm{\rho}_i(k+1) = \bm{\psi}_i \bm{\rho}_i(k) + \bm{\delta}_i \mathbf{u}_i(k),
\end{align}
where $\bm{\psi}_i=\frac{\partial \mathbf{g}_i}{\partial \bm{\rho}_i}|_{\bm{\rho}_i = 0}$ and $\bm{\delta}_i=\frac{\partial \mathbf{g}_i}{\partial \mathbf{u}_i}|_{\mathbf{u}_i = 0}$. The resulting closed-loop matrix $\mathbf{Z}_i = \bm{\psi}_i + \bm{\delta}_i \mathbf{K}_i$ is Schur stable,  $\mathbf{Z}_i \in \mathbb{R}^{(J+1)\times (J+1)}$  \citep{Magni2006Regional}.
\end{asm}

\begin{lmm}
\label{Lemma 2}
If Assumptions \ref{Assumption 1} and \ref{Assumption 2} hold, then the terminal constraint set  $\bm{\mathcal{D}}_{i,T}$ is an invariant set for the closed-loop system in Eq.~(\ref{closed-loop}), and for any $\bm{\rho}_i(k) \in \bm{\mathcal{D}_{i,T}}$, the following inequality always holds,
\begin{align*}
E\left( \bm{\psi}_i \bm{\rho}_i(k) + \bm{\delta}_i \mathbf{u}_i(k)\right)-E\left(\bm{\rho}_i(k)\right) \leq -L_s\left(\bm{\rho}_i(k),\mathbf{K}_i\bm{\rho}_i(k)\right).
\end{align*}
\end{lmm}

\begin{proof}
According to Assumption~\ref{Assumption 2}, for any $\bm{\rho}_i(k) \in \bm{\mathcal{D}_{i,T}}$, the local control law $\mathbf{u}_i(k)=\mathbf{K}_i \bm{\rho}_i(k)$ is feasible, and the state evolution follows $\bm{\rho}_i(k+1) = \mathbf{Z}_i \bm{\rho}_i(k)$, where $\mathbf{Z}_i = \bm{\psi}_i + \bm{\delta}_i \mathbf{K}_i$. Furthermore, ${E\left(\bm{\rho}_i(k)\right)}^{\frac{1}{2}} \leq \varepsilon$ within $\bm{\mathcal{D}_{i,T}}$. Taking the difference of the function $E\left(\bm{\rho}_i(k+1)\right)$, and using Eq.~(\ref{Lyapunov equation}), we have,
\begin{align*}
    &E\left(\bm{\rho}_i(k+1)\right)-E\left(\bm{\rho}_i(k)\right) \\
    &=  E\left(\bm{\psi}_i \bm{\rho}_i(k) + \bm{\delta}_i \mathbf{u}_i(k)\right)-E\left(\bm{\rho}_i(k)\right) \nonumber\\
    &= E\left(\bm{\psi}_i \bm{\rho}_i(k) + \bm{\delta}_i \mathbf{K}_i \bm{\rho}_i(k) \right)-E\left(\bm{\rho}_i(k)\right) \nonumber \\
    &= E\left(\mathbf{Z}_i \bm{\rho}_i(k)\right)-E\left(\bm{\rho}_i(k)\right)  \nonumber \\
    &= {\bm{\rho}_i(k)}^T {\mathbf{Z}_i}^T \mathbf{H}_i \mathbf{Z}_i \bm{\rho}_i(k)-{\bm{\rho}_i(k)}^T \mathbf{H}_i \bm{\rho}_i(k)  \nonumber \\
    &= {\bm{\rho}_i(k)}^T \left({\mathbf{Z}_i}^T \mathbf{H}_i \mathbf{Z}_i-\mathbf{H}_i\right) {\bm{\rho}_i(k)}^T  \nonumber \\
    &=  {\bm{\rho}_i(k)}^T \left[-\left(\mathbf{Q}_i^* + \Delta \mathbf{Q}_i\right)\right] \bm{\rho}_i(k)  \nonumber \\
    &= {\bm{\rho}_i(k)}^T \left[-\left(\mathbf{Q}_i + {\mathbf{K}_i}^T \mathbf{R}_i \mathbf{K}_i + \Delta \mathbf{Q}_i\right)\right] \bm{\rho}_i(k)  \nonumber \\
    &= -{\bm{\rho}_i(k)}^T\mathbf{Q}_i \bm{\rho}_i(k) - {\bm{\rho}_i(k)}^T{\mathbf{K}_i}^T \mathbf{R}_i \mathbf{K}_i \bm{\rho}_i(k)-{\bm{\rho}_i(k)}^T \Delta \mathbf{Q}_i \bm{\rho}_i(k)  \nonumber \\
    &= -L_s\left(\bm{\rho}_i(k), \mathbf{K}_i \bm{\rho}_i(k)\right)-{\bm{\rho}_i(k)}^T \Delta \mathbf{Q}_i \bm{\rho}_i(k)  \nonumber \\
    &\leq -L_s\left(\bm{\rho}_i(k),\mathbf{K}_i \bm{\rho}_i(k)\right).
\end{align*}
Given the positive definiteness of $L_s\left(\bm{\rho}_i(k),\mathbf{K}_i \bm{\rho}_i(k)\right)$, we have $E\left(\bm{\rho}_i(k+1)\right)-E\left(\bm{\rho}_i(k)\right) \leq 0$. Consequently, this implies the condition ${E\left(\bm{\rho}_i(k+1)\right)}^{\frac{1}{2}}\leq {E\left(\bm{\rho}_i(k)\right)}^{\frac{1}{2}} \leq \varepsilon$ holds. According to Definition~\ref{definition 1}, this indicates that the set $\bm{\mathcal{D}_{i,T}}$ is an invariant set for the closed-loop system in Eq.~(\ref{closed-loop}).  
\end{proof}

Next, we demonstrate the recursive feasibility of the optimal control problem in Eq.~(\ref{control}) and the stability of the closed-loop system in Eq.~(\ref{closed-loop}).

\begin{thm}
\label{theorem 1}
If Assumptions \ref{Assumption 1} and \ref{Assumption 2} hold, then for any given $\lambda_i \geq 0$, the optimal control problem in Eq.~(\ref{control}) is recursively feasible within the feasible initial state set $\bm{\mathcal{D}}_N$, and $\bm{\mathcal{D}}_N$ is an invariant set for the closed-loop system in Eq.~(\ref{closed-loop}).
\end{thm}

\begin{proof}
First, based on the state of subsystem $i$ at time step $k-1$, denoted as $\bm{\rho}_i(k-1) \in \bm{\mathcal{D}_{N}}$, and the corresponding optimal control sequence of the optimization problem in Eq.~(\ref{control})  $\bm{\mu}_i^{*}(k-1)=\left[\mathbf{u}_i^{*}(k-1+0),\mathbf{u}_i^{*}(k-1+1),\ldots,\mathbf{u}_i^{*}(k-1+N-1)\right]$, as well as the associated optimal predicted state sequence $\bm{\xi}_i^{*}(k-1)=\left[\bm{\rho}_i^{*}(k-1+1),\bm{\rho}_i^{*}(k-1+2),\ldots,\bm{\rho}_i^{*}(k-1+N)\right]$, we construct a control sequence at time $k$,
\begin{align*}
    \hat{\bm{\mu}}_i(k)=\left[\mathbf{u}_i^{*}(k-1+1),\mathbf{u}_i^{*}(k-1+2),\ldots,\mathbf{u}_i^{*}(k-1+N-1),  \mathbf{K}_i \bm{\rho}_i^{*}(k-1+N)\right].
\end{align*}

Since the terminal state $\bm{\rho}_i^{*}(k-1+N)$  belongs to the terminal set $\bm{\mathcal{D}_{i,T}}$, and the local control law $\mathbf{u}_i^{*}(k-1+N)=\mathbf{K}_i \bm{\rho}_i^{*}(k-1+N)$  ensures that $\max\left(\mathbf{u}_i^{*}(k-1+N)\right) \in U$, the constructed control sequence  $\hat{\bm{\mu}}_i(k)$ satisfies the control constraint (\ref{control-e}). The corresponding predicted state sequence is,
\begin{align*}
    \hat{\bm{\xi}}_i(k)=\left[\bm{\rho}_i^{*}(k-1+2),\bm{\rho}_i^{*}(k-1+3),\ldots,\bm{\rho}_i^{*}(k-1+N), \hat{\bm{\rho}}_i(k+N)\right],
\end{align*}
where $\hat{\bm{\rho}}_i(k+N)=\bm{\psi}_i \bm{\rho}_i^{*}(k-1+N) + \bm{\delta}_i \mathbf{K}_i \bm{\rho}_i^{*}(k-1+N)$.
According to Lemma~\ref{Lemma 2}, $\hat{\bm{\rho}}_i(k+N) \in \bm{\mathcal{D}}_{i,T}$, implying that $ \hat{\bm{\xi}}_i(k)$ satisfies the state constraints, and $E\left(\hat{\bm{\rho}}_i(k+N)\right) - E\left(\bm{\rho}_i^{*}(k-1+N)\right) \leq -L_s\left(\bm{\rho}_i^{*}(k-1+N), \mathbf{K}_i\bm{\rho}_i^{*}(k-1+N)\right)$.
Therefore, the control sequence $\hat{\bm{\mu}}_i(k)$ satisfies all the constraints of the stability optimization problem in Eq.~(\ref{value function optimal}), ensuring the feasibility of the problem at time $k$. 

Assume that the optimal solution of the problem in Eq.~(\ref{value function optimal}) at time $k$ is $\bm{\mu}_i^{v}(k)$, and $V_i\left(\bm{\xi}_i^{v}(k),\bm{\mu}_i^{v}(k)\right) \leq {V}_i\left(\hat{\bm{\xi}}_i(k),\hat{\bm{\mu}}_i(k)\right)$. By comparing the optimal cost values at consecutive time steps, we have,
\begin{align}
\label{contraction}
&  V_i\left(\bm{\xi}_i^{v}(k),\bm{\mu}_i^{v}(k)\right)- V_i\left(\bm{\xi}_i^{*}(k-1),\bm{\mu}_i^{*}(k-1)\right) \nonumber \\
    &\leq {V}_i \left(\hat{\bm{\xi}}_i(k), \hat{\bm{\mu}}_i(k)\right) -V_i\left(\bm{\xi}_i^{*}(k-1),\bm{\mu}_i^{*}(k-1)\right) \nonumber \\
    &=E\left(\hat{\bm{\rho}}_i(k+N)\right) + \sum_{t=1}^{N-1}L_s\left(\bm{\rho}_i^{*}(k-1+t),\mathbf{u}_i^{*}(k-1+t)\right) \nonumber \\ 
    &\phantom{=\;}+L_s\left(\bm{\rho}_i^{*}(k-1+N),\mathbf{K}_i\bm{\rho}_i^{*}(k-1+N)\right)-E\left(\bm{\rho}_i^{*}(k-1+N)\right)\nonumber\\
    &\phantom{=\;} -\sum_{t=0}^{N-1}L_s\left(\bm{\rho}_i^{*}(k-1+t),\mathbf{u}_i^{*}(k-1+t)\right)\nonumber\\
   &=E\left(\hat{\bm{\rho}}_i(k+N)\right)+L_s\left(\bm{\rho}_i^{*}(k-1+N),\mathbf{K}_i\bm{\rho}_i^{*}(k-1+N)\right)\nonumber\\
  &\phantom{=\;}-E\left(\bm{\rho}_i^{*}(k-1+N)\right)-L_s\left(\bm{\rho}_i^{*}(k-1),\mathbf{u}_i^{*}(k-1)\right)\nonumber\\
 &\leq L_s\left(\bm{\rho}_i^{*}(k-1+N),\mathbf{K}_i\bm{\rho}_i^{*}(k-1+N)\right) -L_s\left(\bm{\rho}_i^{*}(k-1),\mathbf{u}_i^{*}(k-1)\right) \nonumber\\
&\phantom{=\;}  -L_s\left(\bm{\rho}_i^{*}(k-1+N),\mathbf{K}_i\bm{\rho}_i^{*}(k-1+N)\right)\nonumber\\
    & =  -L_s\left(\bm{\rho}_i^{*}(k-1),\mathbf{u}_i^{*}(k-1)\right) \nonumber\\
    & \leq 0.
\end{align}
In the contraction function $\eta_i\left(\bm{\xi}_i(k),\lambda_i\right)$, we have $V_i\left(\bm{\xi}_i^{*}(k-1), \bm{\mu}_i^{*}(k-1)\right) - V_i\left(\bm{\xi}_i^{v}(k), \bm{\mu}_i^{v}(k)\right)$ $ \geq 0$, and for any $\lambda_i \geq 0$, it holds that $\eta_i\left(\bm{\xi}_i(k),\lambda_i\right) \geq 0$ and $V_i\left(\bm{\xi}_i^{v}(k), \bm{\mu}_i^{v}(k)\right) \leq \eta_i\left(\bm{\xi}_i(k),\lambda_i\right)$. Therefore, $\bm{\mu}_i^{v}(k)$ is a feasible solution to Eq.~(\ref{control}), demonstrating that the problem exhibits recursive feasibility . 

Consequently,  for any $\bm{\rho}_i(k) \in \bm{\mathcal{D}}_N$, there exists a feasible solution such that the successor state $\bm{\rho}_i(k+1)=\mathbf{g}_i\left(\bm{\rho}_i(k), \mathbf{u}_i(k), \mathbf{e}_i(k)\right) \in \bm{\mathcal{D}}_N$. According to Definitions~\ref{definition 1} and~\ref{definition 3}, the set $\bm{\mathcal{D}}_N$ is positively invariant under the closed-loop system in Eq.~\ref{closed-loop}. 
\end{proof}

\begin{thm}
\label{theorem 2}
If Assumptions \ref{Assumption 1} and \ref{Assumption 2} hold and the optimal control problem in Eq.~(\ref{control}) admits a feasible solution at the initial time, then for any $\lambda_i \in [0,1)$, the closed-loop system in Eq.~(\ref{closed-loop}) is ISS within the invariant set $\bm{\mathcal{D}}_N$.
\end{thm}

\begin{proof}
From the Eq.~(\ref{valuefunction}), we have,
\begin{align*}
    \alpha_{i,1}\left(\left\| \bm{\rho}_i(k)\right\|\right) =q \times {\left\| \bm{\rho}_i(k)\right\|}^2 \leq L_s\left(\bm{\rho}_i(k),\mathbf{u}_i(k)\right) \leq V_i\left(\bm{\xi}_i(k),\bm{\mu}_i(k)\right), 
\end{align*}
where $q$ denotes the minimum eigenvalue of matrix $\mathbf{Q}_i$, and $\alpha_{i,1}\left(\left\| \bm{\rho}_i(k)\right\|\right)$ is a $K_\infty$ class function, which can be regarded as a lower bound of the stability function $V_i\left(\bm{\xi}_i(k),\bm{\mu}_i(k)\right)$ at time step $k$.

To establish an upper bound for the stability function, we define,
\begin{align*}
    &{L}_{s,\max}=\max\Big\{L_s\left(\bm{\rho}_i(k),\mathbf{u}_i(k)\right)|\bm{\rho}_i(k) \in \bm{\mathcal{D}}_N,\max\left(\mathbf{u}_i(k)\right)\in U\Big\}, \\ 
    &{E}_{\max}=\max\Big\{E\left(\bm{\rho}_i(k)\right)|\bm{\rho}_i(k) \in \bm{\mathcal{D}}_N\Big\}.
\end{align*}
Then it follows that,
\begin{align*}
    V_i\left(\bm{\xi}_i(k), \bm{\mu}_i(k)\right) \leq E_{\max}+N L_{s,\max}.
\end{align*}
Based on this, we construct a function,
\begin{align*}
\alpha_{i,2}\left(\left\| \bm{\rho}_i(k)\right\|\right) = \theta \times {\left\| \bm{\rho}_i(k)\right\|}^2 ,
\end{align*}
with a properly chosen parameter $\theta \geq 0$, such that,
\begin{align*}
\theta \times {\left\| \bm{\rho}_i(k)\right\|}^2 \geq E_{\max}+N L_{s,\max}, \quad  \forall \bm{\rho}_i(k) \in \bm{\mathcal{D}}_N.
\end{align*}
This ensures that  $V_i\left(\bm{\xi}_i(k),\bm{\mu}_i(k)\right) \leq \alpha_{i,2}\left(\left\| \bm{\rho}_i(k)\right\|\right)$, and $\alpha_{i,2}\left(\left\| \bm{\rho}_i(k)\right\|\right)$ is a $K_\infty$ class function. 

According to the recursive feasibility established in Theorem~\ref{theorem 1}, a feasible solution exists at every time step $k$. Consider the optimal solutions at two consecutive time steps, $\bm{\mu}_i^{*}(k-1)$ and $\bm{\mu}_i^{*}(k)$.  When  $\lambda_i \in [0,1)$,  from the contraction constraint (\ref{control-g}) and the Eq.~(\ref{contraction}), we obtain,
\begin{align*}
    &V_i\left(\bm{\xi}_i^{*}(k), \bm{\mu}_i^{*}(k)\right)-V_i\left(\bm{\xi}_i^{*}(k-1), \bm{\mu}_i^{*}(k-1)\right) \\
    &\leq (1-\lambda_i)\left[V_i\left(\bm{\xi}_i^{v}(k), \bm{\mu}_i^{v}(k)\right)-V_i\left(\bm{\xi}_i^{*}(k-1), \bm{\mu}_i^{*}(k-1)\right)\right] \\
    &\leq (1-\lambda_i)\left[-L_s\left(\bm{\rho}_i^{*}(k-1), \mathbf{u}_i^{*}(k-1)\right)\right]\\
    &\leq (1-\lambda_i)\left[q \times {\left\| \bm{\rho}_i^{*}(k-1)\right\|}^2\right]\\
    &=-(\lambda_i - 1)\alpha_{i,1}\left(\left\| \bm{\rho}_i^{*}(k-1)\right\|\right)\\
    &=-\alpha_{i,3}\left(\left\| \bm{\rho}_i^{*}(k-1)\right\|\right),
\end{align*}
where $\alpha_{i,3}\left(\left\| \bm{\rho}_i^{*}(k-1)\right\|\right)$ is also a $K_\infty$ class function. Clearly, this implies that,
\begin{align*}
    V_i\left(\bm{\xi}_i^{*}(k), \bm{\mu}_i^{*}(k)\right) \leq V_i\left(\bm{\xi}_i^{*}(k-1), \bm{\mu}_i^{*}(k-1)\right),
\end{align*}
i.e., the value function $V_i\left(\bm{\xi}_i(k), \bm{\mu}_i(k)\right)$ is monotonically decreasing along the closed-loop trajectory. 

In conclusion, the value function $V_i\left(\bm{\xi}_i(k), \bm{\mu}_i(k)\right)$ satisfies the key conditions of Lemma~\ref{Lemma 1}. Therefore, it qualifies as an ISS-Lyapunov function for subsystem $i$, and subsystem $i$ is ISS within the invariant set $\bm{\mathcal{D}}_N$. 
\end{proof}

In the multiagent truncated rollout framework, the optimization problem for each subsystem is solved sequentially while incorporating control inputs from all other CAVs. The terminal subsystem, by aggregating the ISS properties of all constituent subsystems, guarantees the stability of the overall system within the invariant set $\bm{\mathcal{D}}_N$. 
\section{Computational Issues in Rollout Algorithms}
\label{sec_optimize}
In this section, we analyze the computational complexity of the multiagent truncated rollout method and compare it with centralized model predictive control (MPC) and distributed model predictive control (DMPC) approaches (Section~\ref{Time Complexity Analysis}). We also investigate the optimal decision sequence for CAVs (Section~\ref{CAV Decision Order}) and derive theoretical error bounds for the objective function (Section~\ref{Error Bounds}). These analyses establish the foundation for the performance and stability assurances of the proposed control scheme. We finally present the detailed solution algorithms for implementing the proposed strategy in multiagent environments (Section~\ref{Solution Algorithm}).

\subsection{Time Complexity Analysis}
\label{Time Complexity Analysis}
In this subsection, we systematically compare and analyze the computational complexity reduction afforded by this approach in a multiagent system comprising $(I+1)$ CAVs. We assume that each CAV has a prediction and control horizon of $N$ and that its control action space is discretized into a set containing $S$ discrete control actions at any prediction time $k+t$. 

It is known that in centralized MPC, the dynamic behavior of all subsystems is managed by a central controller, which collects system state information and computes the globally optimal control inputs for all CAVs. Although this approach can account for the cooperation of all CAVs and theoretically achieve a global optimum, as the system scale increases, the number of variables and constraints that the central controller must handle grows significantly. In the assumed environment, at each time step, the algorithm is required to predict the control actions of all CAVs over the optimal horizon $N$. Consequently, the computational complexity grows exponentially with both the number of CAVs and the prediction horizon, and is given by $\textit{O}\left(S^{(I+1) \times N}\right)$.

In DMPC, the overall system is decomposed into multiple subsystems, where each subsystem solves its own optimization problem independently. Over the prediction horizon $N$, the complexity for each CAV is $\textit{O}(S^N)$, and the total system complexity can be expressed as $\textit{O}\left(S^N \times (I+1)\right)$. However, the individually optimal solutions of CAVs do not necessarily coincide with the system global optimum, and independent solving leads to poor coordination among agents. To address this, CAVs are allowed multiple rounds of communication, iteratively optimizing their control actions after acquiring updated control information from other CAVs.
Although the iterative optimization introduces some computational delay, in practice, the number of iterations required for the system to converge to an approximately optimal solution is quite small (e.g., $P=3 \ to \ 5$). Moreover, in large-scale systems, computational complexity that grows linearly with $(I+1)$ is significantly lower than complexity that grows exponentially. This enables DMPC to often achieve shorter computation times and better scalability than centralized MPC within a limited number of iterations.

Different from traditional DMPC, the proposed multiagent truncated rollout method not only adopts an agent-by-agent approach to convert the coordination of multiple subsystems into sequential interactions, but also employs truncated rollout during the iteration process to reduce the optimization horizon. As shown in Table~\ref{complexity}, at each iteration, each CAV $i$ optimizes only the first $M_i^{J} \leq N$ steps of its control sequence, reducing the computational complexity for a single vehicle to $\textit{O}\left(S^{M_i^{J}}\right)$. Consequently, the overall computational complexity per iteration is reduced to $\textit{O}\left(\sum_{i=0}^I S^{M_i^{J}}\right)$.

\begin{table}[h]
    \centering
    \renewcommand{\arraystretch}{1.2}
    \caption{Comparison of computational complexity.}
    \label{complexity}  
    \begin{tabular}{l l}
        \hline
        Approach & Time complexity\\
        \hline
        Model predictive control &$O\left(S^{N\times (I+1)}\right)$\\
        Distributed model predictive control &$O\left(S^{N}\times (I+1)\right)$\\
        Multiagent truncated rollout &$\textit{O}\left(\sum_{i=0}^I S^{M_i^{J}}\right)$\\
        \hline
    \end{tabular}
\end{table}

\subsection{Decision Order Optimization}
\label{CAV Decision Order}
In the agent-by-agent sequential decision-making framework, the order in which CAVs make decisions directly impacts the quality of the generated strategies. In this section, we propose a dynamic decision-ordering mechanism that reconstructs the decision sequence in multiagent truncated rollout algorithms by evaluating the potential contribution of each CAV to the overall traffic state at the current time step. This approach aims to improve system-level performance while preserving the cost improvement property of rollout \citep{Dimitri2021}.

Specifically, at each time step $k$, the decision order of the CAVs is determined through the following steps,

(1) For all CAVs $i \in \{0,1,\ldots, I\}$, the optimal objective value $J_i\left(\bm{\xi}_i^{*}(k),\bm{\mu}_i^{*}(k)\right)$ is evaluated in parallel under the assumption that CAV $i$ is the first to optimize, while all other CAVs follow a predefined car-following model. The CAV with the lowest objective value is selected as the first agent to optimize,
\begin{align*}
    i_0^* = \arg \min_{i \in \{0,1,\ldots,I\}} J_i\left(\bm{\xi}_i^{*}(k),\bm{\mu}_i^{*}(k)\right).
\end{align*}

(2) Fix the control result of CAV $i_0^*$, i.e., $\bm{\mu}_{i_0^*}^{*}(k)$, and for the remaining CAVs, compute in parallel their optimal objective values assuming each is the second decision-maker, while the others continue to follow the car-following model. The one with the lowest cost is selected as the second decision-maker.

(3) Repeat the above procedure until a complete decision sequence $i_0^* \to i_1^* \to \ldots \to i_I^*$ is obtained. This also yields an initial joint control input $\bm{\mu}_{i_0^*:i_I^*}^{*}(k)$  for all CAVs at time step $k$.

This dynamic mechanism ensures that CAVs with greater influence on traffic dynamics are prioritized in decision-making, thereby providing more accurate future-state predictions for subsequent vehicles. Although this method requires a total of $(I+1)+I+\ldots+1=(I+1)(I+2)/2$ objective function evaluations, the computational burden is substantially alleviated through multithreaded parallel processing, making the runtime comparable to that of agent-by-agent rollout with a fixed decision order.

\subsection{Error Bounds for the Objective Function}
\label{Error Bounds}
In Eq.~(\ref{M_value_J}) and Eq.~(\ref{M_value_V}), the selection of the truncated horizons $M_i^J$ and $M_i^V$ is related to the upper and lower bounds of the corresponding objective functions. According to Lemma~\ref{Lemma 1} and Theorem~\ref{theorem 2}, there exist two $K_\infty$ class functions, $\alpha_{i,1}\left(\left\| \bm{\rho}_i(k)\right\|\right) =q \times {\left\| \bm{\rho}_i(k)\right\|}^2$ and $\alpha_{i,2}\left(\left\| \bm{\rho}_i(k)\right\|\right) = \theta \times {\left\| \bm{\rho}_i(k)\right\|}^2$, such that the stability objective function $V_i\left(\bm{\xi}_i(k), \bm{\mu}_i(k)\right)$ satisfies,
\begin{align}
\label{V_bound}
\alpha_{i,1}\left(\left\| \bm{\rho}_i(k)\right\|\right) \leq V_i\left(\bm{\xi}_i(k), \bm{\mu}_i(k)\right) \leq \alpha_{i,2}\left(\left\| \bm{\rho}_i(k)\right\|\right).
\end{align}
This implies that the upper and lower bounds of the function $V_i\left(\bm{\xi}_i(k), \bm{\mu}_i(k)\right)$ are determined by $\alpha_{i,1}\left(\left\| \bm{\rho}_i(k)\right\|\right)$ and $\alpha_{i,2}\left(\left\| \bm{\rho}_i(k)\right\|\right)$, respectively. 
Next, we derive similar bounds for the control objective function $J_{i}\left(\bm{\xi}_{i}(k), \bm{\mu}_{i}(k)\right)$.

\begin{prp}
\label{truncated-rollout}
Let $\bm{\mu}_i^{v}(k)$ be the optimal solution of $V_i\left(\bm{\xi}_i(k), \bm{\mu}_i(k)\right)$  at time step $k$ in Eq.~(\ref{value function optimal}). When the jam density of each road cell is $R$, and the set of control actions for CAV $i\in \mathcal{I}(k)$ is $U = [u_{\min}, u_{\max}]$, the optimal value $\bm{\mu}_{i}^{*}(k)$ of the control objective function $J_{i}\left(\bm{\xi}_{i}(k), \bm{\mu}_{i}(k)\right)$ in Eq.~(\ref{control}) satisfies,
\begin{align}
\label{J_bound}
    \frac{\Delta x \Delta t}{q} \left[ V_{i}\left( \bm{\xi}_{i}^{v}(k), \bm{\mu}_{i}^{v}(k) \right) - N r u_{\max}^{2} - (J + 1) h R^{2} \right] \leq J_{i}\left( \bm{\xi}_{i}^{*}(k), \bm{\mu}_{i}^{*}(k) \right) \leq J_{i}\left( \bm{\xi}_{i}^{v}(k), \bm{\mu}_{i}^{v}(k) \right),
\end{align}
where $q$, $r$, and $h$ are the maximum values in the state weighting matrix $\mathbf{Q}_i$, control action weighting matrix $\mathbf{R}_i$, and terminal state weighting matrix $\mathbf{H}_i$, respectively.
\end{prp}
\begin{proof}
First, we explain the upper bound of the $J_{i}\left( \bm{\xi}_{i}^{*}(k), \bm{\mu}_{i}^{*}(k) \right)$. From Eq.~(\ref{contraction}), the optimal solution $\bm{\mu}_{i}^{*}(k-1)$ of the stability objective function at time step $k-1$ satisfies,
\begin{align*}
  V_{i}\left( \bm{\xi}_{i}^{*}(k-1), \bm{\mu}_{i}^{*}(k-1)\right)-
  V_{i}\left( \bm{\xi}_{i}^{v}(k), \bm{\mu}_{i}^{v}(k)\right) \geq 0.
\end{align*}
When $\lambda_i \leq 0$, we have,
\begin{align*}
V_{i}\left( \bm{\xi}_{i}^{v}(k), \bm{\mu}_{i}^{v}(k) \right) \leq V_{i}\left( \bm{\xi}_{i}^{v}(k), \bm{\mu}_{i}^{v}(k) \right) + \lambda_{i} \left[ V_{i}\left( \bm{\xi}_{i}^{*}(k-1), \bm{\mu}_{i}^{*}(k-1) \right) - V_{i}\left( \bm{\xi}_{i}^{v}(k), \bm{\mu}_{i}^{v}(k) \right) \right].
\end{align*}
This implies that $\bm{\mu}_i^{v}(k)$ satisfies constraints~(\ref{control-b})-(\ref{control-g}), and it is a feasible solution to the control problem in Eq.~(\ref{control}), but not necessarily the optimal one. Based on the minimization property of the objective function, we can conclude,
\begin{align*}
    J_{i}\left( \bm{\xi}_{i}^{*}(k), \bm{\mu}_{i}^{*}(k) \right) \leq J_{i}\left( \bm{\xi}_{i}^{v}(k), \bm{\mu}_{i}^{v}(k) \right).
\end{align*}

Next, we prove the lower bound of $J_{i}\left( \bm{\xi}_{i}^{*}(k), \bm{\mu}_{i}^{*}(k) \right)$. From the definition of $V_{i}\left( \bm{\xi}_{i}(k), \bm{\mu}_{i}(k) \right)$ in Eq.~(\ref{valuefunction})-(\ref{E_rho_k}), we obtain,
\begin{align*}
    &V_{i}\left( \bm{\xi}_{i}(k), \bm{\mu}_{i}(k) \right) \\
    &= \sum_{t=0}^{N-1} \left( {\bm{\rho}_{i}(k+t)}^T \mathbf{Q}_{i} \bm{\rho}_{i}(k+t) + {\mathbf{u}_{i}(k+t)}^T \mathbf{R}_{i} \mathbf{u}_{i}(k+t) \right) + {\bm{\rho}_{i}(k+N)}^T \mathbf{H}_{i} \bm{\rho}_{i}(k+N) \\
    &\leq \sum_{t=0}^{N-1} \left( q \sum_{j=0}^{J} {\rho_{i,j}(k+t)}^2 + r \sum_{j=0}^{J} {u_{i,j}(k+t)}^2 \right) + h \sum_{j=0}^{J} {\rho_{i,j}(k+N)}^2.
\end{align*}
Since only the action of the cell where CAV $i$ is located is non-zero in $\mathbf{u}_{i}(k+t)$, and $\max\left(\mathbf{u}_{i}(k+t)\right) \leq u_\max$, 
\begin{align*}
    r \sum_{j=0}^{J} {u_{i,j}(k+t)}^2 \leq r u_{\max}^{2}.
\end{align*}
Similarly, for each cell $j$, $\rho_{i,j}(k+N)\leq R$,
\begin{align*}
    h \sum_{j=0}^{J} {\rho_{i,j}(k+N)}^2 \leq (J + 1) h R^{2}.
\end{align*}
Thus,
\begin{align*}
    V_{i}\left( \bm{\xi}_{i}^{v}(k), \bm{\mu}_{i}^{v}(k) \right) \leq V_{i}\left( \bm{\xi}_{i}(k), \bm{\mu}_{i}(k) \right) \leq q \sum_{t=0}^{N-1} \sum_{j=0}^{J} {\rho_{i,j}(k+t)}^2 + N r u_{\max}^{2} + (J + 1) h R^{2},
\end{align*}
and then,
\begin{align*}
    \sum_{t=0}^{N-1} \sum_{j=0}^{J} {\rho_{i,j}(k+t)}^2 \geq \frac{1}{q} \left[ V_{i}\left( \bm{\xi}_{i}^{v}(k), \bm{\mu}_{i}^{v}(k) \right) - N r u_{\max}^{2} - (J + 1) h R^{2} \right],
\end{align*}
where $q >0$.

Since the vehicle density in each cell $j$ satisfies $\rho_{i,j}(k+t) < 1$, so $\rho_{i,j}(k+t) > {\rho_{i,j}(k+t)}^2$ and 
\begin{align*}
    \sum_{t=0}^{N-1} \sum_{j=0}^{J} \rho_{i,j}(k+t) > \sum_{t=0}^{N-1} \sum_{j=0}^{J} {\rho_{i,j}(k+t)}^2.
\end{align*}
Consequently,
\begin{align*}
    &J_{i}\left( \bm{\xi}_{i}(k), \bm{\mu}_{i}(k) \right) = \sum_{t=0}^{N-1} \sum_{j=0}^{J} \rho_{i,j}(k+t) \Delta x \Delta t \\
    &\geq J_{i}\left( \bm{\xi}_{i}^{*}(k), \bm{\mu}_{i}^{*}(k) \right) \\
    &> \sum_{t=0}^{N-1} \sum_{j=0}^{J} \rho_{i,j}(k+t)^{2} \Delta x \Delta t \\
    &\geq \frac{\Delta x \Delta t}{q} \left[ V_{i}\left( \bm{\xi}_{i}^{v}(k), \bm{\mu}_{i}^{v}(k) \right) - N r u_{\max}^{2} - (J + 1) h R^{2} \right].
\end{align*}
\end{proof}

\subsection{Solution Algorithm for Multiagent Truncated Rollout}
\label{Solution Algorithm}
Building on the previously introduced CAV decision order optimization method and the objective function bounds evaluation mechanism, this subsection outlines the overall algorithmic framework for achieving efficient coordination among multiple CAVs.
We divide the multiagent truncated rollout method into two key components: (1) optimal control sequence solving for individual CAVs, and (2) a sequential iteration scheme for all CAVs.

\begin{algorithm}[H]
  \caption{Multiagent truncated rollout algorithm.}
  \label{algorithm:Rollout-DMPC}
  \begin{algorithmic}[1]
    \Statex \textbf{Input:}
      Total optimization time $T \Delta t$; 
      Termination condition $\epsilon$;
      Prediction horizon $N$;
      Objective function $J_i\left(\bm{\xi}_i(k),\bm{\mu}_i(k)\right)$;
    \Statex \textbf{Output:}
      Total travel time $J_t$ of vehicles on the road during optimization time $T \Delta t$;
    \State Initialize $J_t := 0$;
    \For{$k = 0$ to $T-1$}
        \State Get traffic state $\bm{\rho}(k)$ and construct the set of the CAVs $\mathcal{I}(k)=\left\{1,2,\ldots,I\right\}$;
        \State Set the objective function difference $\Delta:=+\infty$; 
        \State $J_t \leftarrow J_t + \sum_{j=0}^J \left[\rho_j(k) \times \Delta t \times \Delta x\right]$;
        \State Set the iteration count $p := 0$;
        \State Conduct parallel evaluation of the $J_i\left(\bm{\xi}_i(k),\bm{\mu}_i(k)\right)$ for all CAVs;
        \State Generate the decision sequence $\ell(k)$ and initial solution set $\bm{\mu}_{0:I}^{*,0}(k)$;
        \While {$\Delta > \epsilon$}
            \State $p \leftarrow p+1$;
            \ForEach{CAV $i$ in $\ell(k)$}
                \State Calculate $M_i^V$ and solve for $\bm{\mu}_i^{v,p}(k)$;
                \State Calculate $M_i^J$ and solve for $\bm{\mu}_i^{*,p}(k)$;
            \EndFor
            \State $\Delta \leftarrow \left|J_I\left(\bm{\xi}_i^{p}(k),\bm{\mu}_i^{p}(k)\right)-J_I\left(\bm{\xi}_i^{p-1}(k),\bm{\mu}_i^{p-1}(k)\right)\right|$;
        \EndWhile
        \State Execute the first control action $\mathbf{u}_i^{*}(k)$ for all CAV $i\in \mathcal{I}(k)$.
    \EndFor
  \end{algorithmic}
\end{algorithm}

As described in Section~\ref{MPC framework}, at each time step $k\in\{1,2,\ldots\}$, the optimal control sequence for CAV $i \in \mathcal{I}(k)$ is determined by solving the problem defined in Eq.~(\ref{control}). CAV $i$ first minimizes the stability function $V_i\left(\bm{\xi}_i(k),\bm{\mu}_i(k)\right)$ to generate a candidate control sequence $\bm{\mu}_i^{v,p}(k)$. Subsequently, a contraction constraint $\eta_i\left(\bm{\xi}_i(k),\lambda_i\right)$ is constructed to ensure system stability while optimizing the cost objective function, yielding the optimal control sequence $\bm{\mu}_i^{*,p}(k)$. During the iterative process, the optimization horizons $M_i^V$ and $M_i^J$ ($M_i^V \leq N$ and $M_i^J \leq N$) are dynamically adjusted: initialization employs the full horizon $N$, followed by adaptive truncation of the prediction horizon based on evaluating solution quality using the upper and lower bounds of the objective functions ($J_i^{UB}(k)$, $J_i^{LB}(k)$ or $V_i^{UB}(k)$, $V_i^{LB}(k)$). This strategy effectively reduces computational burden without significantly compromising performance.

To coordinate the control decisions of all CAVs within the total optimization horizon $T\Delta t$ while minimizing the cumulative travel time of all vehicles, we propose an iterative scheme as detailed in Algorithm~\ref{algorithm:Rollout-DMPC}. 
At each time step $k$, the impact of individual CAVs on traffic efficiency is evaluated under the current state to optimize the decision-making order and generate initial solutions. Then, CAVs sequentially implement iterative improvements to this initial solution according to the fixed optimization order.
It is important to note that when the time step $k+N$ exceeds $T$, both the prediction and control horizons of the controlled CAV are adjusted to $T-k$. 
\section{Numerical Results}
\label{sec_results}
In this section, we present numerical experiments to validate the proposed method. We outline the experimental settings (Section~\ref{Settings}) and evaluate the performance of the multiagent truncated rollout through comparisons with several state-of-the-art baseline algorithms (Section~\ref{Performance Comparison}). To provide further insights, we then apply the method to diverse traffic scenarios and conduct sensitivity analyses (Section~\ref{Sensitivity Analysis}). 
Finally, we analyze the scalability of the multiagent truncated rollout approach on large-scale transportation systems (Section~\ref{Scalability_Analysis}).

\subsection{Experimental Settings}
\label{Settings}
We select a real-world bottleneck section of the Hujin Highway in Shanghai, China, as the case study, utilizing the simulation of urban mobility (SUMO) platform \citep{Lopez2018SUMO} for the simulation environment. As shown in Fig.~\ref{Hujin Highway Experiments}, the model features a 3000 $m$ segment with a lane reduction from three to two lanes. This includes a 2100 $m$ upstream coordination zone and a 900 $m$ warm-up segment to allow vehicles to stabilize into realistic driving behaviors before entering the control area.
In addition to average travel time, we also introduce average waiting time as a critical metric for evaluating traffic flow stability. In bottleneck scenarios, waiting time serves as a direct proxy for the intensity of stop-and-go waves and merging friction. Furthermore, human-driven vehicles (HDVs) are assumed to follow the classic intelligent driver model (IDM) \citep{TreiberIDM2000}.

\begin{figure}[htbp]
  \centering
  \subfigure[Real section of a single bottleneck]
  {\label{hujin-highway}
  \includegraphics[width=0.9\textwidth,trim=0 0 0 0, clip]{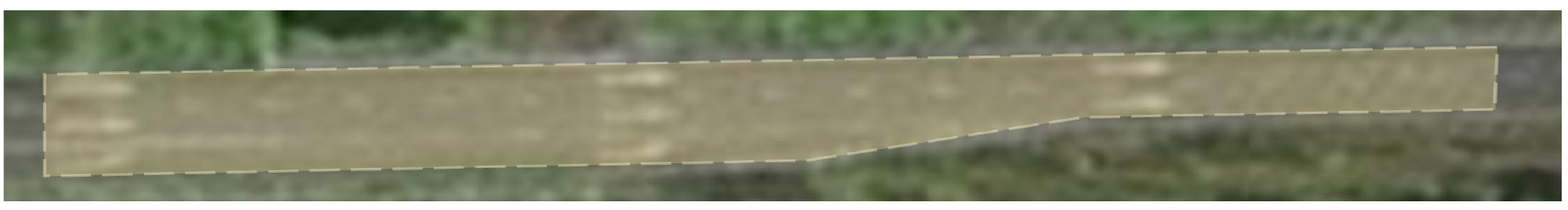}}
  \subfigure[Simulation of a single bottleneck]
  {\label{hujin-sumo}
  \includegraphics[width=0.9\textwidth,trim=0 0 0 0, clip]{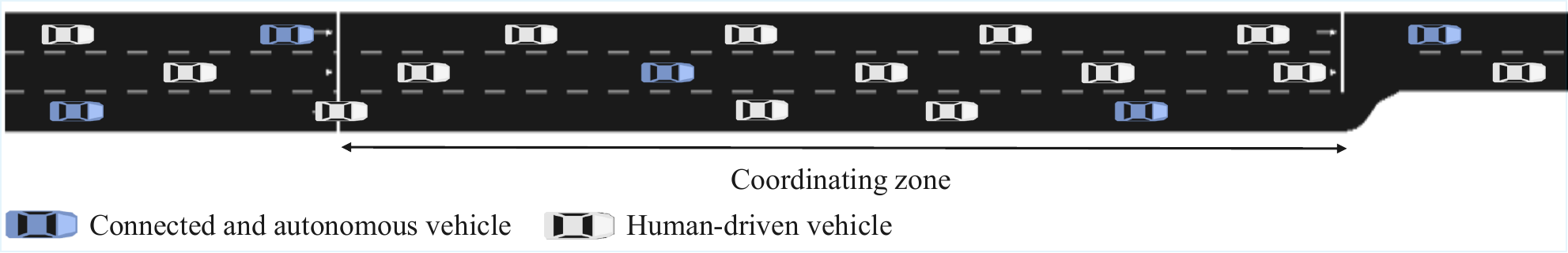}} 
  \caption{A single bottleneck section on the Hujin Highway.} 
  \label{Hujin Highway Experiments}
\end{figure}

To emulate a high-density traffic environment, vehicles are generated at the input of the experimental section (Fig.~\ref{hujin-sumo}) with a demand of 2000 $veh/h$. Lane assignment is randomized, and the CAV penetration rate is set to 15\%. The coordination zone is spatially discretized into multiple interconnected cells, each 300 $m$ in length. Regarding the controller settings, the state and control weighting matrices for each CAV $i$ are defined as $Q_i=\text{diag}(0.5, \dots, 0.5) \in \mathbb{R}^{(J+1)\times (J+1)}$ and $R_i=\text{diag}(1.0, \dots, 1.0) \in \mathbb{R}^{(J+1)\times (J+1)}$, respectively. Detailed parameter settings are provided in Table~\ref{Parameter}.

\begin{table}[h]
    \centering
    \caption{Parameter settings for the experiment.}
    \label{Parameter}  
    \begin{tabular}{l l}
        \hline
        Parameters & Values\\
        \hline
        MPC prediction horizon $N$ &7 $s$\\
        Cell length $\Delta x$ &300 $m$\\
        Jam density $R$ & 0.12 $veh/m$\\
        Free flow speed $V$ & 33.33 $m/s$\\
        Minimum control speed $u_{\min}$ &5 $m/s$\\
        Maximum control speed $u_{\max}$ &33.33 $m/s$\\
        Total road length &3000 $m$\\
        Constriction factor $\lambda_i$ &0.6\\
        \hline 
    \end{tabular}
\end{table}

\subsection{Performance Comparison}
\label{Performance Comparison}
We apply the multiagent truncated rollout method method for longitudinal coordination control of CAVs, and compare it with five commonly used methods:

\begin{itemize}
    \item Without control: This approach serves as the baseline where all vehicles (including CAVs) follow the IDM in SUMO simulations, devoid of external control interventions.
    
    \item Multiagent proximal policy optimization (MAPPO): MAPPO is a typical multiagent algorithm in RL \citep{Yu2021TheSE}, which achieves stable learning and efficient cooperation of agents by implementing centralized learning and decentralized execution.
    
    \item Heterogeneous-agent proximal policy optimization (HAPPO): The method enhances MAPPO via agent-by-agent sequential decision-making \citep{Liu2024AMR}, where each agent's policy update incorporates preceding agents' optimized strategies.

    \item Model predictive control (MPC): It optimizes the control sequence for a finite horizon and uses a feedback correction mechanism to correct the prediction error based on the predictive model, achieving optimal control of the system \citep{ZHANG2023199partb}. 
    
    \item Distributed model predictive control (DMPC): The approach decomposes the global system into decentralized subsystems, each of which is governed by a separate controller to reduce computational complexity \citep{Goulet2022DMPC}.
\end{itemize}

As shown in Table \ref{Performance-comparisons}, the multiagent truncated rollout method achieves concurrent improvements in average travel time and waiting time, leveraging its precise modeling of global dynamics and sequential agent-by-agent decision-making. Compared to the without control scenario, the average travel time decreases by 17.29\%. Most notably, the reduction in waiting time is substantial, reaching 33.86\%. These results demonstrate that although the proposed method primarily focuses on longitudinal speed regulation, it creates ample maneuvering space by smoothing traffic flow and eliminating stop-and-go phenomena. This facilitates smoother lane changing at bottlenecks, thereby mitigating the traffic friction induced by mandatory lane changes.

\begin{table}[h]
    \centering
    \caption{Performance comparisons of different approaches.}
    \label{Performance-comparisons}
    \begin{tabular}{l l l}
        \hline
        Approach & Average travel time & Average waiting time\\
        \hline
        Without control & 191.54 $s$ & 37.57 $s$\\
        MAPPO & 169.95 $s$ ($\downarrow$ 11.27$\%$) & 32.56 $s$ ($\downarrow$ 13.33$\%$)\\
        HAPPO & 168.39 $s$ ($\downarrow$ 12.09$\%$) & 30.71 $s$ ($\downarrow$ 18.26$\%$)\\
        MPC & 165.94 $s$ ($\downarrow$ 13.37$\%$) & 29.26 $s$ ($\downarrow$ 22.12$\%$)\\
        DMPC & 166.20 $s$ ($\downarrow$ 13.23$\%$) & 26.19 $s$ ($\downarrow$ 30.29$\%$)\\
        Our approach & \textbf{158.43 $s$} (\textbf{$\downarrow$ 17.29$\%$}) &\textbf{24.85 $s$} (\textbf{$\downarrow$ 33.86$\%$})\\
        \hline
    \end{tabular}
    \par \vspace{0.5em}
    \begin{minipage}{\linewidth}
        \footnotesize \textit{Note: The presented results are the average values derived from 5 independent simulation runs.}
    \end{minipage}
\end{table}

These results also indicate that, compared to the conventional MAPPO algorithm, the sequential decision-making mechanism enables HAPPO to operate in a more stationary environment during training, thereby yielding superior control performance. However, the efficacy of MARL coordination strategies is inherently constrained by the fixed agent population size defined during training. Dynamic fluctuations in agent numbers often necessitate network retraining, rendering these methods ill-suited for real-time control. Consequently, RL-based methods underperform compared to MPC-based approaches in our simulation environment.
Although centralized MPC theoretically offers global optimality, its implementation in large-scale multiagent systems faces two primary limitations. First, as the number of CAVs rises, the optimization problem becomes highly non-convex and high-dimensional, making it computationally prohibitive to find high-quality solutions within tight time windows. Second, the solver frequently gets trapped in local optima. These constraints limit the method to an average travel time reduction of only 13.37\%.
Similarly, although DMPC leverages distributed control to achieve a 13.23\% reduction in travel time, it suffers from performance degradation caused by information transmission latency during iterative optimization. This results in a performance gap of 4.06\% compared to the proposed multiagent truncated rollout method.

\subsection{Sensitivity Analysis}
\label{Sensitivity Analysis}
In this section, we explore the sensitivity of the proposed method to two key parameters: CAV penetration rate and coordination zone length, using the scenario in Fig.~\ref{Hujin Highway Experiments}. To ensure the statistical reliability of the results, each experimental scenario is executed over 5 independent simulation runs. The results presented in the subsequent figures depict the average values, with shaded regions representing the standard deviation, thereby visually demonstrating the robustness of the control performance against traffic uncertainties.

\subsubsection{CAV Penetration Rate}
\label{CAV Penetration}
\begin{figure}[htbp]
  \centering
  \subfigure[Average travel time with the penetration rate]
  {\label{Sensitivity_Traveltime}
  \includegraphics[width=0.45\textwidth,trim=0 0 0 0, clip]{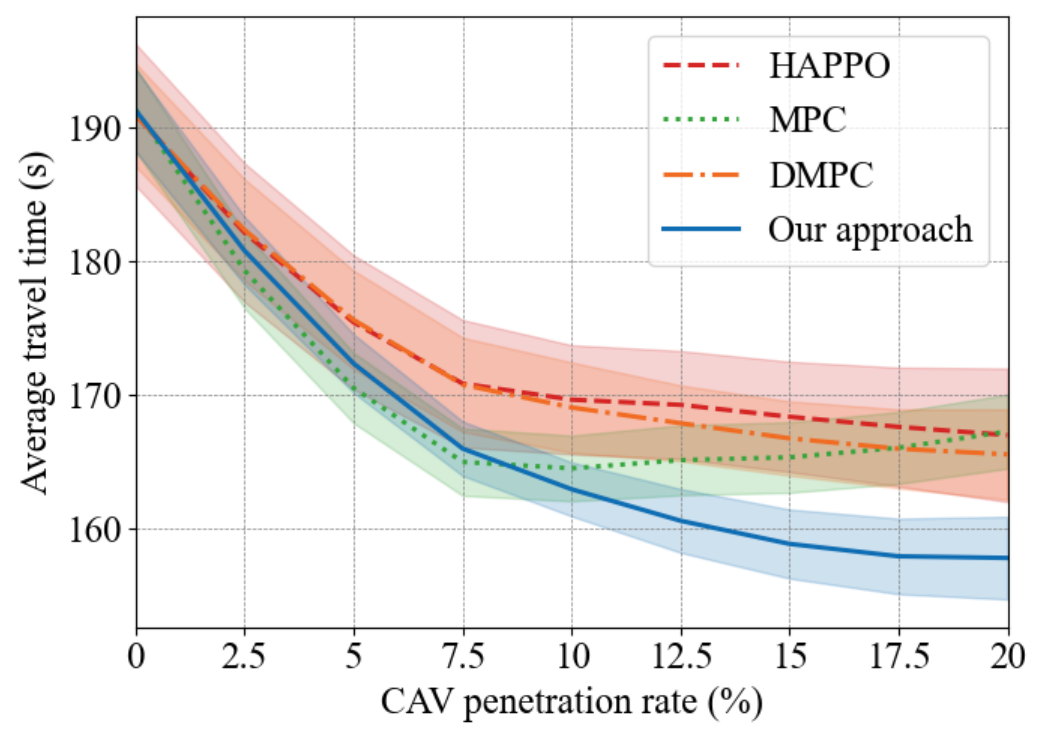}}
  \subfigure[Average speed with the penetration rate]
  {\label{Sensitivity_Speed}
  \includegraphics[width=0.45\textwidth,trim=0 0 0 0, clip]{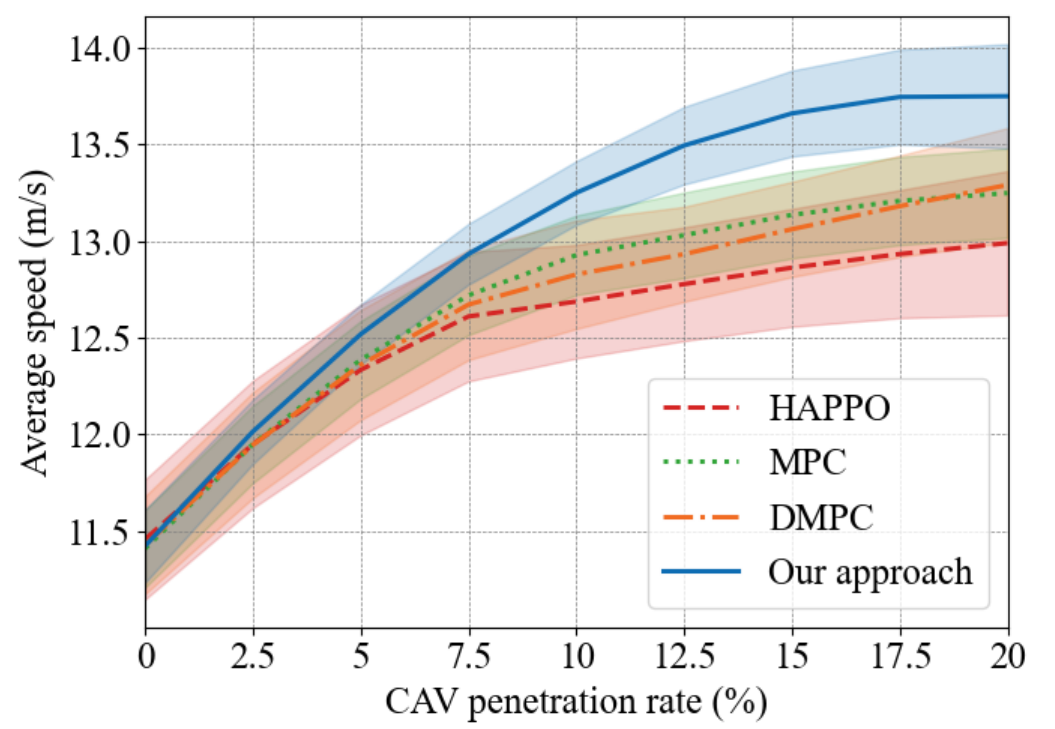}} 
  \caption{Sensitivity analysis of CAV penetration rate impact on system performance.} 
  \label{Sensitivity_penetration}
\end{figure}

We first analyze the impact of CAV penetration rate on average travel time and average speed. As shown in Fig. \ref{Sensitivity_Traveltime}, the average travel time for all methods decreases as CAV penetration increases. Specifically, at low penetration rates (below 8\%), MPC outperforms other baselines due to its precise solving capability for small-scale optimization problems. However, when the penetration rate exceeds 10\%, the rapid increase in dimensionality and coupling complexity makes the centralized MPC prone to getting trapped in local optima, preventing it from finding high-quality solutions within acceptable computation times. Consequently, its control performance degrades significantly.
While DMPC benefits from distributed control at high penetration rates, its inherent local optimization nature limits its global effectiveness. Similarly, the HAPPO algorithm underperforms when facing uncertainties in multiagent environments due to convergence challenges. These trends are further corroborated by the average speed results in Fig. \ref{Sensitivity_Speed}.
When the CAV penetration rate surpasses 15\%, the escalating complexity of cooperative interactions impedes the scalability of HAPPO, MPC, and DMPC, resulting in suboptimal performance. In contrast, the proposed multiagent truncated rollout approach demonstrates superior scalability in high-penetration scenarios, reducing average travel time by over 16\% and increasing average speed by more than 18\%. Notably, the control benefits of CAVs begin to saturate at approximately 15\% penetration. This saturation suggests that the critical stop-and-go waves have been effectively suppressed.

\subsubsection{Coordination Zone Length}
Next, we examine the impact of coordination zone length on system performance. As illustrated in Fig.~\ref{Sensitivity_zone}, extending the coordination zone from 300 $m$ to 2400 $m$ results in a non-monotonic trend where the average travel time initially decreases before subsequently rising across all methods. This phenomenon stems from the trade-off between enhanced collaboration potential and increased optimization complexity. In shorter zones, the constrained maneuvering space limits the cooperative capabilities of CAVs. Conversely, expanding the coordination range allows for earlier trajectory planning and better coordination, reducing average travel time by up to 18\%. However, when the zone length exceeds 1800 $m$, although the expanded state space improves decision-making foresight, it drastically escalates the computational burden, leading to performance degradation.
Comparative analysis reveals that while the multiagent truncated rollout method slightly trails centralized MPC in short zones (though still outperforming DMPC), it exhibits superior scalability in larger coordination areas ($>900$ $m$). By achieving lower travel times than MPC in these extended ranges, the proposed method demonstrates its robustness and adaptability to large-scale systems.

\begin{figure}[htbp]
  \centering
  \subfigure[Average travel time with the zone length]
  {\label{zone_Traveltime}
  \includegraphics[width=0.45\textwidth,trim=0 0 0 0, clip]{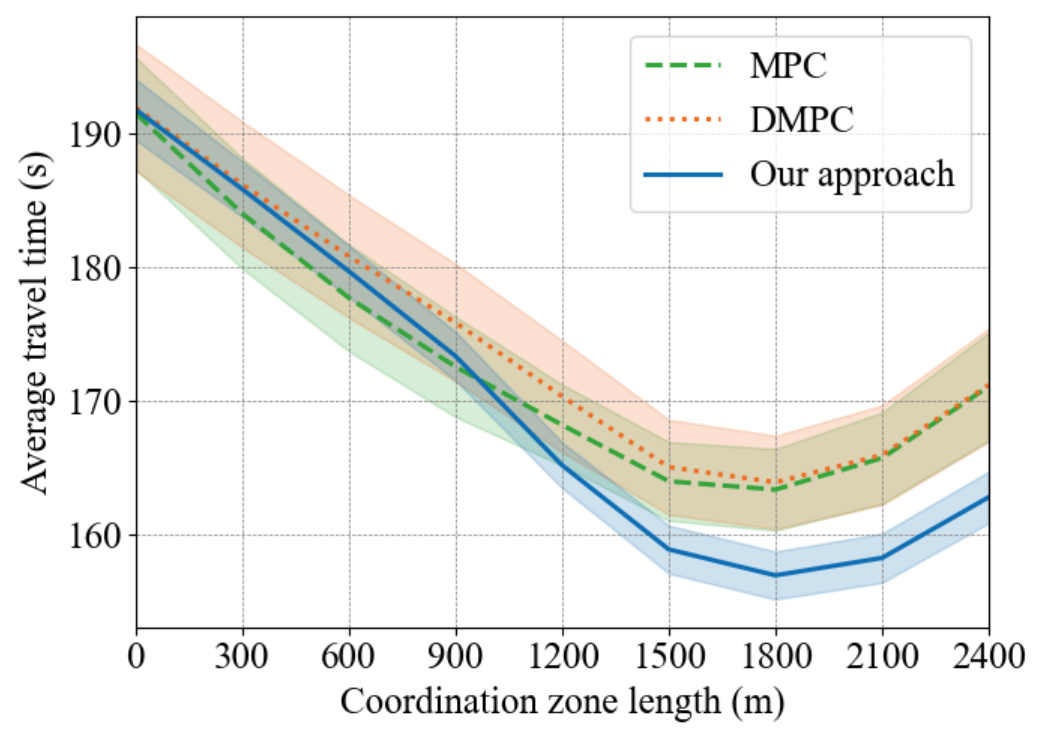}}
  \subfigure[Average speed with the zone length]
  {\label{zone_Speed}
  \includegraphics[width=0.45\textwidth,trim=0 0 0 0, clip]{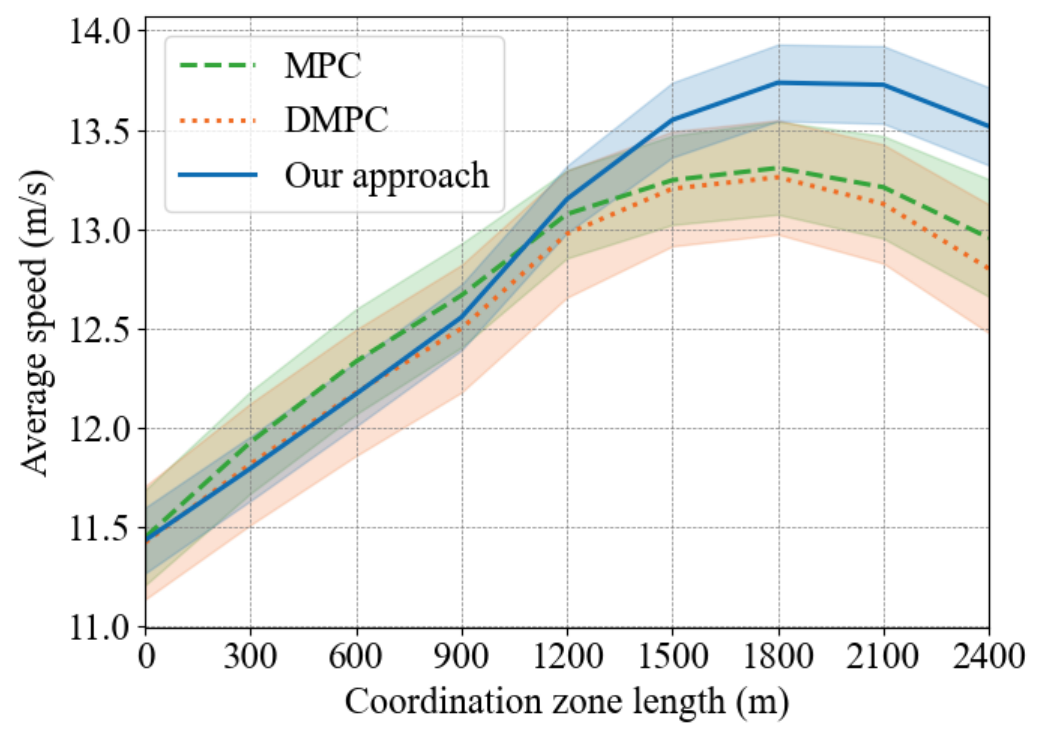}} 
  \caption{Sensitivity analysis of coordination zone length impact on system performance.} 
  \label{Sensitivity_zone}
\end{figure}

\subsection{Scalability Analysis}
\label{Scalability_Analysis}
To evaluate the scalability of the multiagent truncated rollout method in large-scale scenarios, we conducted a series of extended experiments within the 2100 $m$ control zone (Section~\ref{Settings}), focusing on computation time under varying CAV quantities. Similarly, each data point represents the average of 5 independent simulation runs. The negligible standard deviation observed confirms the stability of the solvers across different traffic states.
In addition to the MPC and DMPC benchmarks, we introduced a standard multiagent rollout baseline. 
While this method shares the distributed, agent-by-agent sequential framework with our proposed approach, it fundamentally differs in horizon management: it utilizes a fixed prediction horizon ($N=7$), standing in sharp contrast to the dynamic truncation mechanism of our proposed method.

We first recorded the average single-step coordination time for 1 to 8 CAVs in identical scenarios. Subsequently, to comprehensively validate computational efficiency, we measured total computation time over 1000 $s$ simulations with CAV penetration rates ranging from 10\% to 40\%. 
As illustrated in Fig.~\ref{computation_time}, while computational costs increase with the number of CAVs for all methods, significant algorithmic disparities exist.

\begin{figure}[H]
  \centering
  \subfigure[Single time step computation time]
  {\label{single_step_computation_time}
  \includegraphics[width=0.65\textwidth,trim=0 0 0 0, clip]{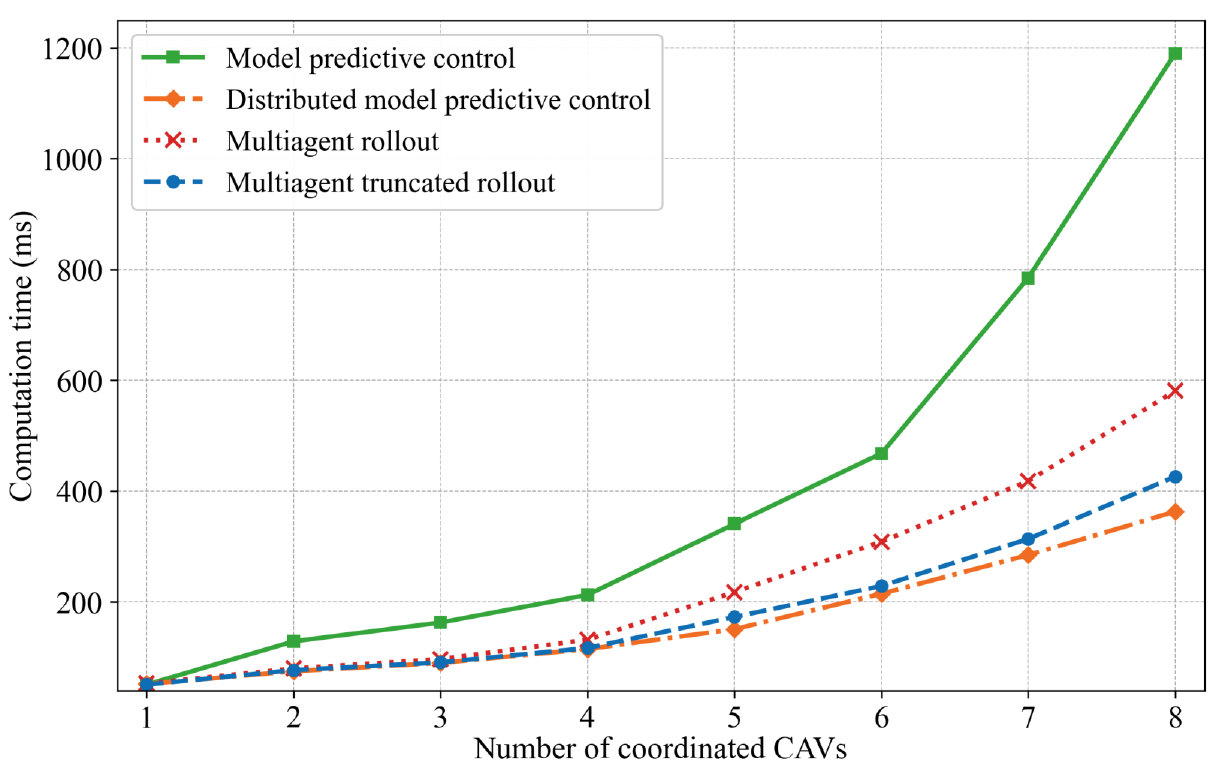}}
  \subfigure[Total system computation time]
  {\label{system_computation_time}
  \includegraphics[width=0.65\textwidth,trim=0 0 0 0, clip]{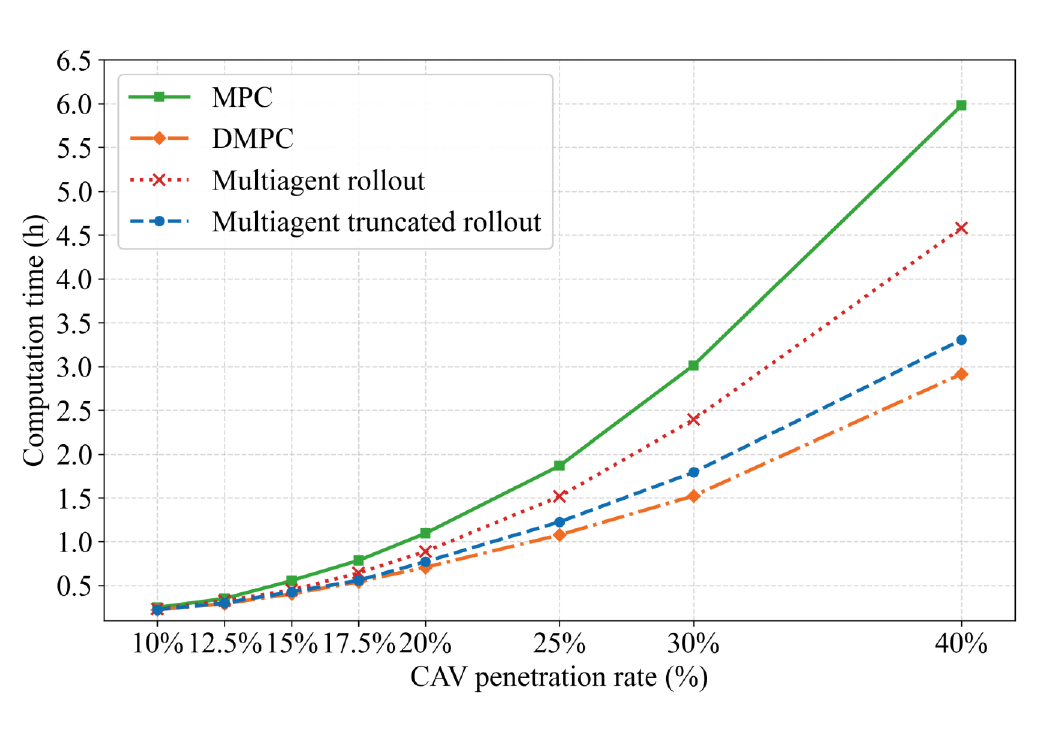}} 
  \caption{Computational scalability of control methods.} 
  \label{computation_time}
\end{figure}

Centralized MPC faces severe bottlenecks when coordinating numerous CAVs, exhibiting the steepest increase in solution time. The standard multiagent rollout incurs higher costs than DMPC, as its sequential nature inherently limits parallelization. However, the proposed truncated mechanism significantly reduces the dimensionality of decision variables, ensuring real-time feasibility even in high-penetration systems.
Crucially, sensitivity analysis (Section~\ref{CAV Penetration}) indicates that traffic regulation benefits saturate at approximately 15\% penetration. At this operating point, the proposed method reduces computation time by 22.9\% compared to MPC and is only 5.7\% higher than DMPC, all while achieving superior control performance. These findings confirm that the proposed method effectively balances high performance with computational efficiency, thereby resolving scalability concerns in practical deployments.

The computational efficiency improvement enabled by the truncated rollout mechanism is further validated across varying coordination zone lengths in the case study depicted in Fig.~\ref{Hujin Highway Experiments}. As indicated by the red curve in Fig.~\ref{improvement}, under a 15\% CAV penetration rate, coordination zones shorter than 600 $m$ yield only a marginal 2\% reduction in computation time compared to the standard rollout method. This minimal gain is attributed to the limited optimization scope, where the complexity differential between the two approaches is negligible. However, as the coordination zone expands, the efficiency gains from the truncated mechanism become significantly more pronounced. Specifically, when the zone length exceeds 2400 $m$, the total system computation time is reduced by more than 10\%, highlighting the method's advantage in larger-scale scenarios.

\begin{figure}[h]
    \centering
    \includegraphics[width=0.6\textwidth]{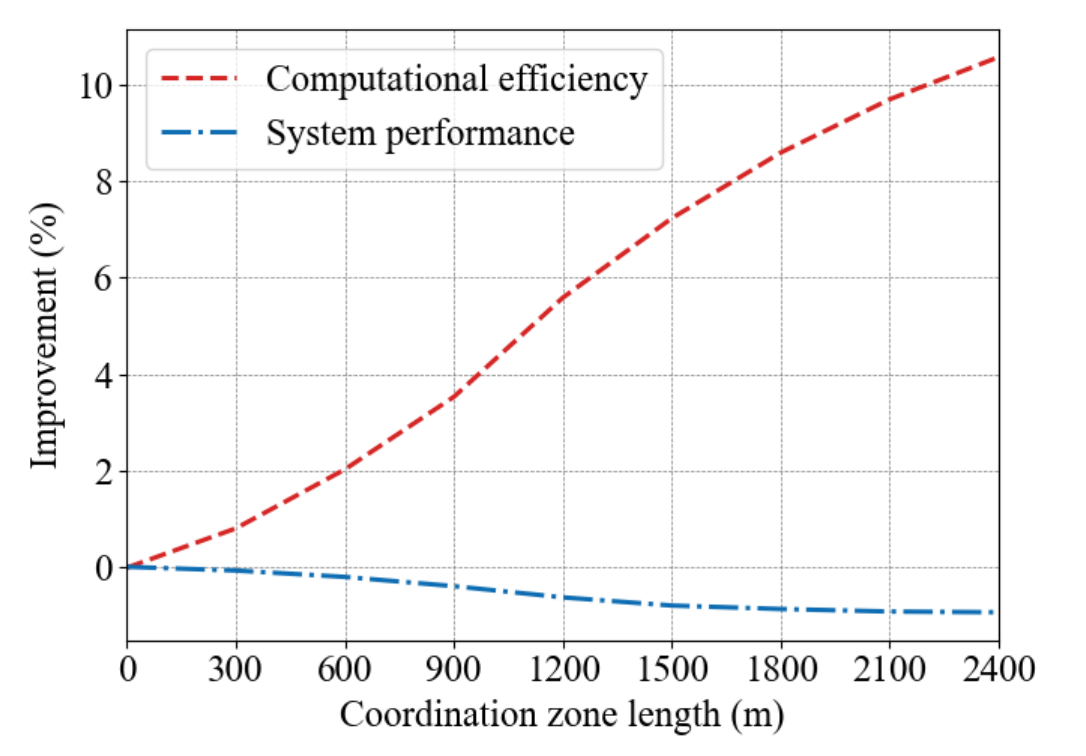}
    \caption{Improvement in computational efficiency and system performance.}
    \label{improvement}
\end{figure}

While the truncated rollout mechanism delivers substantial efficiency improvements, it incurs a marginal degradation in control performance compared to the standard rollout baseline due to the partial horizon. However, as illustrated by the blue curve in Fig.~\ref{improvement}, this performance gap remains consistently below 1\%. For instance, within the optimal 1800 $m$ control zone, the average travel time for the truncated method is 157.01 $s$ versus 155.65 $s$ for the standard method—a difference of merely 1.36 $s$ (0.87\%). Given the corresponding reduction in computation time of over 8\%, this minor performance compromise is considered a highly favorable trade-off for real-time implementation. 
\section{Conclusion}
\label{sec_conclude}
This paper proposes a novel multiagent truncated rollout method for optimizing highway traffic flow in mixed autonomy, with a specific focus on congested bottleneck scenarios. 
Based on a coupled PDE-ODE model that captures the hybrid dynamics of macroscopic flow and microscopic vehicle behaviors, we established a unified road density evolution equation. This served as the foundation for a distributed coordination framework where each CAV functions as an independent controller. 
Unlike traditional distributed model predictive control that relies on parallel optimization, our architecture introduces an agent-by-agent sequential optimization mechanism. This enables explicit coordination by allowing CAVs to incorporate the latest decisions of neighboring agents. 
Furthermore, the integration of an adaptive truncated rollout technique significantly reduces computational costs by dynamically adjusting the optimization horizon based on rigorous cost function bounds. 
Theoretical analysis guarantees the input-to-state stability of individual subsystems and the monotonic improvement of the system-level objective. 
Simulation results calibrated to the Shanghai Hujin Expressway demonstrate that the proposed method achieves a superior balance between control performance and scalability compared to conventional model predictive control approaches. Crucially, although this study focuses on longitudinal speed control, the results indicate that proactively regulating upstream traffic density creates favorable time gaps for merging vehicles. This effectively facilitates the mandatory lane changes required at the bottleneck, thereby mitigating traffic congestion even in the absence of explicit lateral control.

This work can be extended in several directions. 
First, we aim to explicitly incorporate lane-changing decision-making to achieve the joint optimization of longitudinal and lateral dynamics. This will provide a more comprehensive representation of real-world driving behaviors. 
Second, the impact of uncertainties, including sensor noise, communication delays, and the stochasticity of human-driven vehicles, will be rigorously analyzed. Future iterations will introduce robust optimization or learning-based adaptive mechanisms to enhance system reliability. 
Finally, we plan to scale the control framework to larger networks by designing spatially distributed coordination mechanisms capable of managing cascading congestion across multiple interconnected bottlenecks. 
\appendix
\section*{Appendix}
\label{Appendix}
\renewcommand{\thesubsection}{Appendix \Alph{subsection}}
\subsection{Derivation of the Density Transfer Equation}
\label{Appendix A}
\subsubsection*{Cells Without CAV Bottlenecks}
In road cells without CAV bottlenecks, the calculation of density follows Eq.~\ref{supply-demand}. We organize it as,
\begin{align*}
    \rho_j(k+1) = \rho_j(k) + f_1\Big(\rho_{j-1}(k),\rho_{j}(k),\rho_{j+1}(k)\Big),
\end{align*}
where the nonlinear function $f_1\left(\rho_{j-1}(k),\rho_{j}(k),\rho_{j+1}(k)\right)$ is derived from Eqs.~(\ref{supply-demand-b})-(\ref{supply-demand-e}),
\begin{align*}
    f_1\Big(\rho_{j-1}(k),\rho_{j}(k),\rho_{j+1}(k)\Big)=&-\frac{\Delta t}{\Delta x} \Big[\min\Big\{f\left(\min\left\{\rho_j(k),\rho_*\right\}\right),f\left(\max\left\{\rho_{j+1}(k),\rho_*\right\}\right)\Big\}\\
   & -\min\Big\{f\left(\min\left\{\rho_{j-1}(k),\rho_*\right\}\right),f\left(\max\left\{\rho_{j}(k),\rho_*\right\}\right)\Big\}\Big]
\end{align*}

\subsubsection*{Cells with CAV Bottlenecks}
If the action $u_i(k)$ of CAV $i$ does not satisfy constraint (\ref{CAV reduction}), we reconstruct the upstream density $\hat{\rho}_{j-\frac{1}{2}}\left(u_i(k)\right)$ and downstream density $\check{\rho}_{j+\frac{1}{2}}\left(u_i(k)\right)$ of the cell $j$ in which this CAV is located by solving the following equation,
\begin{align*}
    \frac{\alpha R}{4V} {\left(V-u_i(k)\right)}^2 = V \rho \left(u_i(k)\right) \left(1-\frac{\rho \left(u_i(k)\right)}{R}\right) - u_i(k)\rho \left(u_i(k)\right).
\end{align*}
Then,
\begin{align*}
    \hat{\rho}_{j-\frac{1}{2}}\left(u_i(k)\right)=R\left(V-u_i(k)\right)\frac{1+\sqrt{1-\alpha}}{2V},\\
    \check{\rho}_{j+\frac{1}{2}}\left(u_i(k)\right)=R\left(V-u_i(k)\right)\frac{1-\sqrt{1-\alpha}}{2V}.
\end{align*}

For $F_{j+\frac{1}{2}}(k)$ in Eq.~\ref{CAV_down},
\begin{enumerate}
\item[(i)] If $\Delta t_{i,j}(k) \leq \Delta t$,
\begin{align*}
& F_{j+\frac{1}{2}}(k) = \frac{\Delta t_{i,j}(k)}{\Delta t}f\left( \check{\rho}_{j+\frac{1}{2}}(u_i(k))\right)+\left(1- \frac{\Delta t_{i,j}(k)}{\Delta t}\right)f\left( \hat{\rho}_{j^-\frac{1}{2}}(u_i(k))\right)\\
& =\frac{\Delta t_{i,j}(k)}{\Delta t} \left[f\left( \check{\rho}_{j+\frac{1}{2}}(u_i(k))\right) - f\left( \hat{\rho}_{j-\frac{1}{2}}(u_i(k))\right)\right] +f\left( \hat{\rho}_{j-\frac{1}{2}}(u_i(k))\right).
\end{align*}
Since,
\begin{align*}
    \Delta t_{i,j}(k) = \frac{1-d_{i,j}(k)}{u_i(k)} \Delta x, \quad
    d_{i,j}(k) = \frac{\rho_{j}(k)-\hat{\rho}_{j-\frac{1}{2}}(u_i(k))}{\check{\rho}_{j+\frac{1}{2}}(u_i(k))-\hat{\rho}_{j-\frac{1}{2}}(u_i(k))},
\end{align*}
we can get,
\begin{align*}
&\frac{\Delta t_{i,j}(k)}{\Delta t} \left[f\left( \check{\rho}_{j+\frac{1}{2}}(u_i(k))\right) - f\left( \hat{\rho}_{j-\frac{1}{2}}(u_i(k))\right)\right] \\
&=\frac{\Delta t_{i,j}(k)}{\Delta t} \left[V  \check{\rho}_{j+\frac{1}{2}}(u_i(k))\left(1-\frac{ \check{\rho}_{j+\frac{1}{2}}(u_i(k))}{R}\right) - V  \hat{\rho}_{j-\frac{1}{2}}(u_i(k))\left(1-\frac{ \hat{\rho}_{j-\frac{1}{2}}(u_i(k))}{R}\right)\right]\\
&=\frac{1-d_{i,j}(k)}{u_i(k)}  \frac{\Delta x}{\Delta t} \left[V  \check{\rho}_{j+\frac{1}{2}}(u_i(k))\left(1-\frac{ \check{\rho}_{j+\frac{1}{2}}(u_i(k))}{R}\right) - V  \hat{\rho}_{j-\frac{1}{2}}(u_i(k))\left(1-\frac{ \hat{\rho}_{j-\frac{1}{2}}(u_i(k))}{R}\right)\right]\\
&=\frac{\check{\rho}_{j+\frac{1}{2}}(u_i(k))-\rho_{j}(k)}{\check{\rho}_{j+\frac{1}{2}}(u_i(k))-\hat{\rho}_{j-\frac{1}{2}}(u_i(k))}  \frac{\Delta x}{u_i(k)\Delta t} \Big[V  \check{\rho}_{j+\frac{1}{2}}(u_i(k))\left(1-\frac{ \check{\rho}_{j+\frac{1}{2}}(u_i(k))}{R}\right)\\
 & \quad - V  \hat{\rho}_{j-\frac{1}{2}}(u_i(k))\left(1-\frac{ \hat{\rho}_{j-\frac{1}{2}}(u_i(k))}{R}\right)\Big]\\
&=\frac{\Delta x}{u_i(k)\Delta t} \left[u_i \left( \check{\rho}_{j+\frac{1}{2}}(u_i(k))-\rho_j(k)\right)\right],
\end{align*}
so,
\begin{align*}
 F_{j+\frac{1}{2}}(k) =\frac{\Delta x}{\Delta t} \left( \check{\rho}_{j+\frac{1}{2}}(u_i(k))-\rho_j(k)\right) + f\left( \hat{\rho}_{j-\frac{1}{2}}(u_i(k))\right).
\end{align*}

\item[(ii)] If $\Delta t_{i,j}(k) > \Delta t$,
\begin{align*}
 F_{j+\frac{1}{2}}(k) =f\left( \check{\rho}_{j+\frac{1}{2}}(u_i(k))\right).
\end{align*}
\end{enumerate}

For $F_{j-\frac{1}{2}}(k)$ in Eq.~\ref{CAV_up},
\begin{enumerate}
\item[(i)] If $\rho_{j-1}(k) \leq \rho_*$ and $\hat{\rho}_{j-\frac{1}{2}}(u_i(k)) \leq \rho_*$,
\begin{align*}
& D \left(\rho_{j-1}(k) \right)=f \left(\min \left\{\rho_{j-1}(k) , \rho_* \right\} \right)=f\left(\rho_{j-1}(k) \right),\\
& S \left(\hat{\rho}_{j-\frac{1}{2}}(u_i(k)) \right)=f \left(\max \left\{\hat{\rho}_{j-\frac{1}{2}}(u_i(k)), \rho_* \right\} \right)=f\left(\rho_* \right),\\
&F_{j-\frac{1}{2}}(k) = \min \left \{ D \left(\rho_{j-1}(k) \right), S\left(\hat{\rho}_{j-\frac{1}{2}}(u_i(k))\right) \right\} = f\left(\rho_{j-1}(k) \right).
\end{align*}

\item[(ii)] If $\rho_{j-1}(k) \leq \rho_*$ and $\hat{\rho}_{j-\frac{1}{2}}(u_i(k)) > \rho_*$,
\begin{align*}
& D \left(\rho_{j-1}(k) \right)=f \left(\min \left\{\rho_{j-1}(k) , \rho_* \right\} \right)=f\left(\rho_{j-1}(k) \right),\\
&S \left(\hat{\rho}_{j-\frac{1}{2}}(u_i(k)) \right)=f \left(\max \left\{\hat{\rho}_{j-\frac{1}{2}}(u_i(k)), \rho_* \right\} \right)=f\left(\hat{\rho}_{j-\frac{1}{2}}(u_i(k)) \right),\\
&F_{j-\frac{1}{2}}(k) = \min \left \{ D \left(\rho_{j-1}(k) \right), S\left(\hat{\rho}_{j-\frac{1}{2}}(u_i(k))\right) \right\} = \min \left \{ f\left(\rho_{j-1}(k) \right),f\left(\hat{\rho}_{j-\frac{1}{2}}(u_i(k)) \right)\right\}.
\end{align*}

\item[(iii)] If $\rho_{j-1}(k) > \rho_*$ and $\hat{\rho}_{j-\frac{1}{2}}(u_i(k)) \leq \rho_*$,
\begin{align*}
& D \left(\rho_{j-1}(k) \right)=f \left(\min \left\{\rho_{j-1}(k) , \rho_* \right\} \right)=f\left(\rho_* \right),\\
&S \left(\hat{\rho}_{j-\frac{1}{2}}(u_i(k)) \right)=f \left(\max \left\{\hat{\rho}_{j-\frac{1}{2}}(u_i(k)), \rho_* \right\} \right)=f\left(\rho_* \right),\\
&F_{j-\frac{1}{2}}(k) = \min \left \{ D \left(\rho_{j-1}(k) \right), S\left(\hat{\rho}_{j-\frac{1}{2}}(u_i(k))\right) \right\} =f\left(\rho_* \right).
\end{align*}

\item[(iv)] If $\rho_{j-1}(k) > \rho_*$ and $\hat{\rho}_{j-\frac{1}{2}}(u_i(k)) > \rho_*$,
\begin{align*}
& D \left(\rho_{j-1}(k) \right)=f \left(\min \left\{\rho_{j-1}(k) , \rho_* \right\} \right)=f\left(\rho_* \right),\\
&S \left(\hat{\rho}_{j-\frac{1}{2}}(u_i(k)) \right)=f \left(\max \left\{\hat{\rho}_{j-\frac{1}{2}}(u_i(k)), \rho_* \right\} \right)=f\left(\hat{\rho}_{j-\frac{1}{2}}(u_i(k)) \right),\\
&F_{j-\frac{1}{2}}(k) = \min \left \{ D \left(\rho_{j-1}(k) \right), S\left(\hat{\rho}_{j-\frac{1}{2}}(u_i(k))\right) \right\} =f\left(\hat{\rho}_{j-\frac{1}{2}}(u_i(k)) \right).
\end{align*}
\end{enumerate}

Above all, the density state transfer equation for a road cell containing a CAV bottleneck can be expressed as follows,
\begin{align*}
    &\rho_j(k+1) = \rho_j(k) - \frac{\Delta t}{\Delta x} \left(F_{j+\frac{1}{2}}(k) -F_{j-\frac{1}{2}}(k) \right)\\
    &=\rho_j(k)+d\rho_j(k)+\Big[bc+a(1-b)c\Big]\frac{\Delta t}{\Delta x} f\left(\rho_{j-1}(k)\right)+b(1-c)\frac{\Delta t}{\Delta x} f\left(\rho_*\right)\\
     &\quad  +\Big[ac(1-b)(1-2d)-bd+(1-b)(1-c)(1-d)\Big]\frac{\Delta t}{\Delta x} f\left(\hat{\rho}_{j-\frac{1}{2}}\left(u_i(k)\right)\right)\\
    &\quad  -d\check{\rho}_{j+\frac{1}{2}}\left(u_i(k)\right)-(1-d)\frac{\Delta t}{\Delta x}f\left(\check{\rho}_{j+\frac{1}{2}}\left(u_i(k)\right)\right),
\end{align*}
where $a,b,c,d$ are all Heaviside functions,
\begin{align*}
&a=\begin{cases}
    1, &  f \left(\rho_{j-1}(k) \right) \leq f\left(\hat{\rho}_{j-\frac{1}{2}}(u_i(k)) \right), \\
    0, &f \left(\rho_{j-1}(k) \right) > f\left(\hat{\rho}_{j-\frac{1}{2}}(u_i(k)) \right),
     \end{cases}
&&b=\begin{cases}
    1, &  \hat{\rho}_{j-\frac{1}{2}}(u_i(k)) \leq\rho_*, \\
    0, &\hat{\rho}_{j-\frac{1}{2}}(u_i(k)) >\rho_*,
     \end{cases}\\
&c=\begin{cases}
    1, &  \rho_{j-1}(k) \leq \rho_*, \\
    0, & \rho_{j-1}(k) > \rho_*,
     \end{cases}
&&d=\begin{cases}
    1, &  \Delta t_{i,j}(k) \leq \Delta t, \\
    0, & \Delta t_{i,j}(k) > \Delta t.
     \end{cases}\\
\end{align*}
Then the nonlinear functions $f_2\left(\rho_{j-1}(k),\rho_{j}(k)\right)$ and $f_3\left(u_i(k)\right)$ are,
\begin{align*}
    &f_2\left(\rho_{j-1}(k),\rho_{j}(k)\right)=d\rho_j(k)+\Big[bc+a(1-b)c\Big]\frac{\Delta t}{\Delta x} f\left(\rho_{j-1}(k)\right)+b(1-c)\frac{\Delta t}{\Delta x} f\left(\rho_*\right),\\
    &f_3\left(u_i(k)\right)=\Big[ac(1-b)(1-2d)-bd+(1-b)(1-c)(1-d)\Big]\frac{\Delta t}{\Delta x} f\left(\hat{\rho}_{j-\frac{1}{2}}\left(u_i(k)\right)\right)\\
    &\quad \quad \quad \quad -d\check{\rho}_{j+\frac{1}{2}}\left(u_i(k)\right)-(1-d)\frac{\Delta t}{\Delta x}f\left(\check{\rho}_{j+\frac{1}{2}}\left(u_i(k)\right)\right).
\end{align*}

\subsection{Derivation of Action Range Violating Eq.~(\ref{CAV reduction})}
\label{Appendix B}
At any time step $k$, for the action $u_i(k)$ of CAV $i \in I(k)$ not to satisfy constraint~(\ref{CAV reduction}), there must be $u_i(k) > v\left(\rho_j(k)\right)$ and the following inequality holds,
\begin{align*}
    &f\left(\rho_j(k)\right) - u_i(k)\rho_j(k) -\frac{\alpha R}{4V} {\left(V-u_i(k)\right)}^2\\
    &=V\rho_j(k)\left(1-\frac{\rho_j(k)}{R}\right)-u_i(k)\rho_j(k) -\frac{\alpha R}{4V} {\left(V-u_i(k)\right)}^2\\
    &>0,
\end{align*}
i.e.,
\begin{align*}
    &-V\rho_j(k)\left(1-\frac{\rho_j(k)}{R}\right)+u_i(k)\rho_j(k) +\frac{\alpha R}{4V} {\left(V-u_i(k)\right)}^2\\
    &=\frac{\alpha R}{4V}{u_i(k)^2}-\left(\frac{\alpha R}{2}-\rho_j(k)\right)u_i(k)-\left(V\rho_j(k)-\frac{V{\rho_j(k)}^2}{R}-\frac{\alpha RV}{4}\right)\\
    &< 0,
\end{align*}
where $\rho_j(k)$ denotes the density of cell $j$ where cav $i$ is located.

Solving the quadratic equation $\frac{\alpha R}{4V}{u_i(k)^2}-\left(\frac{\alpha R}{2}-\rho_j(k)\right)u_i(k)-\left(V\rho_j(k)-\frac{V{\rho_j(k)}^2}{R}-\frac{\alpha RV}{4}\right)=0$, we obtain the roots,
\begin{align*}
    \gamma_1\left(\rho_j(k)\right)=V-\frac{2V\rho_j(k)}{\alpha R} \left(1+\sqrt{1-\alpha}\right), \quad \gamma_2\left(\rho_j(k)\right)=V-\frac{2V\rho_j(k)}{\alpha R} \left(1-\sqrt{1-\alpha}\right).
\end{align*}
where $\gamma_2\left(\rho_j(k)\right)$ must be less than $v\left(\rho_j(k)\right)$.
This is because if $\gamma_2\left(\rho_j(k)\right) \leq v\left(\rho_j(k)\right)$ and $\alpha=(W-1)/{W}$, there must be 
\begin{align*}
    \frac{2}{\alpha}\left(1-\sqrt{1-\alpha}\right)=\frac{2W-2\sqrt{W}}{W-1}\leq 1,
\end{align*}
then $W=1$. This is inconsistent with the assumption that the number of lanes upstream of the merge point is $W > 1$.

Since $(\alpha R)/({4V}) > 0$, the solution set of the inequality is,
\begin{align*}
    \gamma_1\left(\rho_j(k)\right) \leq u_i \leq \gamma_2\left(\rho_j(k)\right).
\end{align*}

\bibliographystyle{plainnat}
\bibliography{optimization}

@article{M.L.2014Scalar,
    title = {Scalar conservation laws with moving constraints arising in traffic flow modeling: An existence result},
    journal = {Journal of Differential Equations},
    volume = {257},
    number = {11},
    pages = {4015-4029},
    year = {2014},
    issn = {0022-0396},
    author = {Maria Laura Delle Monache and Paola Goatin}
}

@INPROCEEDINGS{Mladen2018Traffic,
  author={Čičić, Mladen and Johansson, Karl Henrik},
  booktitle={2018 21st International Conference on Intelligent Transportation Systems}, 
  title={Traffic regulation via individually controlled automated vehicles: a cell transmission model approach}, 
  year={2018},
  volume={},
  number={},
  pages={766-771}
}

@article{Giulia2018Traffic,
    title = {Traffic control via moving bottleneck of coordinated vehicles},
    journal = {IFAC-PapersOnLine},
    volume = {51},
    number = {9},
    pages = {13-18},
    year = {2018},
    issn = {2405-8963},
    author = {Giulia Piacentini and Paola Goatin and Antonella Ferrara}
}

@article{xiong2021Optimizing,
    title = {Optimizing coordinated vehicle platooning: An analytical approach based on stochastic dynamic programming},
    journal = {Transportation Research Part B: Methodological},
    volume = {150},
    pages = {482-502},
    year = {2021},
    issn = {0191-2615},
    author = {Xi Xiong and Junyi Sha and Li Jin}
}

@INPROCEEDINGS{Liu2024AMR,
  author={Liu, Lu and Wang, Maonan and Pun, Man-On and Xiong, Xi},
  booktitle={2024 IEEE 27th International Conference on Intelligent Transportation Systems}, 
  title={A Multi-Agent Rollout Approach for Highway Bottleneck Decongestion in Mixed Autonomy}, 
  year={2024},
  volume={},
  number={},
  pages={3377-3382}
}

@article{Bao2023ASB,
    author = {Bao, Jingjing and Wu, Celimuge and Lin, Yangfei and Zhong, Lei and Chen, Xianfu and Yin, Rui},
    year = {2023},
    month = {},
    title = {A scalable approach to optimize traffic signal control with federated reinforcement learning},
    volume = {13},
    pages={19184},
    journal = {Scientific Reports},
    issn = {2045-2322}
}

@article{Qin2021Lighthill,
  title={Lighthill-Whitham-Richards Model for Traffic Flow Mixed with Cooperative Adaptive Cruise Control Vehicles},
  author={ Qin, Yanyan  and  Wang, Hao  and  Ni, Daiheng },
  journal={Transportation science},
  number={4},
  volume={55},
  pages={883-907},
  year={2021}
}

@article{QIU2023Cooperative,
    title = {Cooperative trajectory control for synchronizing the movement of two connected and autonomous vehicles separated in a mixed traffic flow},
    journal = {Transportation Research Part B: Methodological},
    volume = {174},
    pages = {102769},
    year = {2023},
    issn = {0191-2615},
    author = {Jiahua Qiu and Lili Du}
}

@article{Richards1956ShockWO,
  title={Shock Waves on the Highway},
  author={Paul I. Richards},
  journal={Operations Research},
  year={1956},
  volume={4},
  number={1},
  pages={42-51}
}

@article{Bertsekas2020ConstrainedMR,
  title={Constrained Multiagent Rollout and Multidimensional Assignment with the Auction Algorithm},
  author={Dimitri P. Bertsekas},
  journal={ArXiv},
  year={2020},
  volume={abs/2002.07407}
}

@INPROCEEDINGS{Daini2022Centralized,
  author={Daini, Chiara and Goatin, Paola and Monache, Maria Laura Delle and Ferrara, Antonella},
  booktitle={2022 European Control Conference}, 
  title={Centralized Traffic Control via Small Fleets of Connected and Automated Vehicles}, 
  year={2022},
  volume={},
  number={},
  pages={371-376}
}

@ARTICLE{Luo2022Multiobjective,
  author={Luo, Jie and He, Defeng and Zhu, Wei and Du, Haiping},
  journal={IEEE Transactions on Intelligent Transportation Systems}, 
  title={Multiobjective Platooning of Connected and Automated Vehicles Using Distributed Economic Model Predictive Control}, 
  year={2022},
  volume={23},
  number={10},
  pages={19121-19135}
}

@ARTICLE{Magni2006Regional,
  author={Lalo Magni and Davide Raimondo and Riccardo Scattolini},
  journal={IEEE Transactions on Automatic Control}, 
  title={Regional Input-to-State Stability for Nonlinear Model Predictive Control}, 
  year={2006},
  volume={51},
  number={9},
  pages={1548-1553}
}

@article{Bertsekas1997RolloutAF,
  title={Rollout Algorithms for Combinatorial Optimization},
  author={Dimitri P. Bertsekas and John N. Tsitsiklis and Cynara Wu},
  journal={Journal of Heuristics},
  year={1997},
  volume={3},
  pages={245-262}
}

@article{Yu2021TheSE,
  title={The Surprising Effectiveness of MAPPO in Cooperative, Multi-Agent Games},
  author={Chao Yu and Akash Velu and Eugene Vinitsky and Yu Wang and Alexandre M. Bayen and Yi Wu},
  journal={ArXiv},
  year={2021},
  volume={abs/2103.01955}
}

@article{Zu2018RealtimeET,
  title={Real-time energy-efficient traffic control via convex optimization},
  author={Yue Zu and Chenhui Liu and Ran Dai and Anuj Sharma and Jing Dong},
  journal={Transportation Research Part C: Emerging Technologies},
  year={2018},
  volume={92},
  pages={119-136}
}

@article{Liard2023APM,
  title={A PDE-ODE model for traffic control with autonomous vehicles},
  author={Thibault Liard and Raphael E. Stern and Maria Laura Delle Monache},
  journal={Networks Heterog. Media},
  year={2023},
  volume={18},
  number={3},
  pages={1190-1206}
}

@INPROCEEDINGS{Piacentini2019HighwayTC,
  author={Giulia Piacentini and Antonella Ferrara and Ioannis Papamichail and Markos Papageorgiou},
  booktitle={2019 IEEE 58th Conference on Decision and Control}, 
  title={Highway Traffic Control with Moving Bottlenecks of Connected and Automated Vehicles for Travel Time Reduction}, 
  year={2019},
  volume={},
  number={},
  pages={3140-3145}
}

@ARTICLE{Zhu2022Oper,
  author={Zhu, Lian and Lu, Linjun and Wang, Xianing and Jiang, Chenming and Ye, Nanfei},
  journal={IEEE Transactions on Intelligent Transportation Systems}, 
  title={Operational Characteristics of Mixed-Autonomy Traffic Flow on the Freeway With On- and Off-Ramps and Weaving Sections: An RL-Based Approach}, 
  year={2022},
  volume={23},
  number={8},
  pages={13512-13525}
}

@article{Bertsekas2019MultiagentRA,
  title={Multiagent Rollout Algorithms and Reinforcement Learning},
  author={Dimitri Panteli Bertsekas},
  journal={ArXiv},
  year={2019},
  volume={abs/1910.00120}
}

@ARTICLE{Bhattacharya2024MultiagentTRL,
  author={Bhattacharya, Sushmita and Kailas, Siva and Badyal, Sahil and Gil, Stephanie and Bertsekas, Dimitri P.},
  journal={IEEE Transactions on Robotics}, 
  title={Multiagent Reinforcement Learning: Rollout and Policy Iteration for POMDP With Application to Multirobot Problems}, 
  year={2024},
  volume={40},
  number={},
  pages={2003-2023}
}

@article{LI2025103154,
title = {Real-time vehicle relocation, personnel dispatch and trip pricing for carsharing systems under supply and demand uncertainties},
journal = {Transportation Research Part B: Methodological},
volume = {193},
pages = {103154},
year = {2025},
issn = {0191-2615},
author = {Mengjie Li and Haoning Xi and Chi Xie and Zuo-Jun Max Shen and Yifan Hu}
}

@article{ANUPRIYA2023103726,
title = {Congestion in cities: Can road capacity expansions provide a solution?},
journal = {Transportation Research Part A: Policy and Practice},
volume = {174},
pages = {103726},
year = {2023},
issn = {0965-8564},
author = { Anupriya and Prateek Bansal and Daniel J. Graham}
}

@ARTICLE{Gao2022Optimal,
  author={Gao, Zhibo and Wu, Zhizhou and Hao, Wei and Long, Keke and Byon, Young-Ji and Long, Kejun},
  journal={IEEE Transactions on Intelligent Transportation Systems}, 
  title={Optimal Trajectory Planning of Connected and Automated Vehicles at On-Ramp Merging Area}, 
  year={2022},
  volume={23},
  number={8},
  pages={12675-12687}
}

@INPROCEEDINGS{Nie2021speed,
  author={Nie, Wendi and You, Yonghong and Lee, Victor Chung-Sing and Duan, Yaoxin},
  booktitle={2021 IEEE 24th International Conference on Intelligent Transportation Systems}, 
  title={Variable Speed Limit Control for Individual Vehicles on Freeway Bottlenecks with Mixed Human and Automated Traffic Flows}, 
  year={2021},
  volume={},
  number={},
  pages={2492-2498}
}

@article{NTOUSAKIS2016464,
title = {Optimal vehicle trajectory planning in the context of cooperative merging on highways},
journal = {Transportation Research Part C: Emerging Technologies},
volume = {71},
pages = {464-488},
year = {2016},
issn = {0968-090X},
author = {Ioannis A. Ntousakis and Ioannis K. Nikolos and Markos Papageorgiou}
}

@article{CHEN2023103264,
title = {Coordinated traffic control of urban networks with dynamic entrance holding for mixed CAV traffic},
journal = {Transportation Research Part E: Logistics and Transportation Review},
volume = {178},
pages = {103264},
year = {2023},
issn = {1366-5545},
author = {Xiangdong Chen and Xi Lin and Qiang Meng and Meng Li}
}

@article{MA2023104266,
title = {A speed-maximization trajectory optimization model on a reservation-based intersection control system},
journal = {Transportation Research Part C: Emerging Technologies},
volume = {154},
pages = {104266},
year = {2023},
issn = {0968-090X},
author = {Muting Ma and Zhixia Li}
}

@article{Vinitsky2023,
author = {Vinitsky, Eugene and Lichtl\'{e}, Nathan and Parvate, Kanaad and Bayen, Alexandre},
title = {Optimizing Mixed Autonomy Traffic Flow with Decentralized Autonomous Vehicles and Multi-Agent Reinforcement Learning},
year = {2023},
issue_date = {April 2023},
publisher = {Association for Computing Machinery},
address = {New York, NY, USA},
volume = {7},
number = {2},
issn = {2378-962X},
journal = {ACM Transactions on Cyber-Physical Systems},
articleno = {13},
numpages = {22}
}

@article{ZHANG2001337PartB,
title = {A finite difference approximation of a non-equilibrium traffic flow model},
journal = {Transportation Research Part B: Methodological},
volume = {35},
number = {4},
pages = {337-365},
year = {2001},
issn = {0191-2615},
author = {Hongjun Michael Zhang}
}

@article{CHUNG200782PartB,
title = {Relation between traffic density and capacity drop at three freeway bottlenecks},
journal = {Transportation Research Part B: Methodological},
volume = {41},
number = {1},
pages = {82-95},
year = {2007},
issn = {0191-2615},
author = {Koohong Chung and Jittichai Rudjanakanoknad and Michael J. Cassidy}
}

@article{John2024TRR,
author = {John H. Kodi and Emmanuel Kidando and Priyanka Alluri and Thobias Sando},
title ={Guidelines for Activating Ramp Metering Signals in Response to Non-Recurrent Congestion during Off-Peak Hours Using a Statistical Method},
journal = {Transportation Research Record},
volume = {2678},
number = {12},
pages = {1237-1251},
year = {2024}
}

@article{JIN20121000PartB,
title = {A kinematic wave theory of multi-commodity network traffic flow},
journal = {Transportation Research Part B: Methodological},
volume = {46},
number = {8},
pages = {1000-1022},
year = {2012},
issn = {0191-2615},
author = {Wen-Long Jin}
}

@INPROCEEDINGS{Lopez2018SUMO,
  author={Lopez, Pablo Alvarez and Behrisch, Michael and Bieker-Walz, Laura and Erdmann, Jakob and Flötteröd, Yun-Pang and Hilbrich, Robert and Lücken, Leonhard and Rummel, Johannes and Wagner, Peter and Wiessner, Evamarie},
  booktitle={2018 IEEE 21st International Conference on Intelligent Transportation Systems}, 
  title={Microscopic Traffic Simulation using SUMO}, 
  year={2018},
  volume={},
  number={},
  pages={2575-2582}
}

@INPROCEEDINGS{Ferramosca2009nonlinearMPC,
  author={Ferramosca, Antonio and Limon, Daniel and Alvarado, Ignacio and Alamo, Teodoro and Camacho, Eduardo F.},
  booktitle={Proceedings of the 48h IEEE Conference on Decision and Control held jointly with 2009 28th Chinese Control Conference}, 
  title={MPC for tracking of constrained nonlinear systems}, 
  year={2009},
  volume={},
  number={},
  pages={7978-7983}
}

@ARTICLE{Dimitri2021,
  author={Bertsekas, Dimitri P.},
  journal={IEEE/CAA Journal of Automatica Sinica}, 
  title={Multiagent Reinforcement Learning: Rollout and Policy Iteration}, 
  year={2021},
  volume={8},
  number={2},
  pages={249-272}
}

@INPROCEEDINGS{Wang2023ITSCrollout,
  author={Wang, Ning and Chen, Xiao and Mårtensson, Jonas},
  booktitle={2023 IEEE 26th International Conference on Intelligent Transportation Systems}, 
  title={Rollout-Based Interactive Motion Planning for Automated Vehicles}, 
  year={2023},
  volume={},
  number={},
  pages={4187-4194}
}

@article{SHI2021279,
title = {Constructing a fundamental diagram for traffic flow with automated vehicles: Methodology and demonstration},
journal = {Transportation Research Part B: Methodological},
volume = {150},
pages = {279-292},
year = {2021},
issn = {0191-2615},
author = {Xiaowei Shi and Xiaopeng Li}
}

@inproceedings{greenshields1935study,
  title={A study of traffic capacity},
  author={Greenshields, Bruce D and Bibbins, J Rowland and Channing, WS and Miller, Harvey H},
  booktitle={Highway research board proceedings},
  volume={14},
  number={1},
  pages={448--477},
  year={1935},
  organization={Washington, DC}
}

@article{piacentini2020traffic,
  title={Traffic control via platoons of intelligent vehicles for saving fuel consumption in freeway systems},
  author={Piacentini, Giulia and Goatin, Paola and Ferrara, Antonella},
  journal={IEEE Control Systems Letters},
  volume={5},
  number={2},
  pages={593--598},
  year={2020},
  publisher={IEEE}
}

@article{ZHANG2023199partb,
title = {Platoon-centered control for eco-driving at signalized intersection built upon hybrid MPC system, online learning and distributed optimization part II: Theoretical analysis},
journal = {Transportation Research Part B: Methodological},
volume = {172},
pages = {199-216},
year = {2023},
issn = {0191-2615},
author = {Hanyu Zhang and Lili Du}
}

@article{TreiberIDM2000,
author = {Treiber, Martin and Hennecke, Ansgar and Helbing, Dirk},
year = {2000},
month = {02},
pages = {1805-1824},
title = {Congested Traffic States in Empirical Observations and Microscopic Simulations},
volume = {62},
journal = {Physical Review E}
}

@ARTICLE{Goulet2022DMPC,
  author={Goulet, Nathan and Ayalew, Beshah},
  journal={IEEE Transactions on Intelligent Transportation Systems}, 
  title={Distributed Maneuver Planning With Connected and Automated Vehicles for Boosting Traffic Efficiency}, 
  year={2022},
  volume={23},
  number={8},
  pages={10887-10901}
}

@article{DAGANZO1994269,
title = {The cell transmission model: A dynamic representation of highway traffic consistent with the hydrodynamic theory},
journal = {Transportation Research Part B: Methodological},
volume = {28},
number = {4},
pages = {269-287},
year = {1994},
issn = {0191-2615},
author = {Carlos F. Daganzo}
}

@article{ZOU20257,
title = {Dyna-Style Learning with a Macroscopic Model for Vehicle Platooning in Mixed-Autonomy Traffic},
journal = {IFAC-PapersOnLine},
volume = {59},
number = {8},
pages = {7-12},
year = {2025},
note = {5th IFAC Workshop on Control of Systems Governed by Partial Differential Equations - CPDE 2025},
issn = {2405-8963},
author = {Yichuan Zou and Yi Gao and Xi Xiong and Li Jin}
}

@article{LU202226partb,
title = {Are autonomous vehicles better off without signals at intersections? A comparative computational study},
journal = {Transportation Research Part B: Methodological},
volume = {155},
pages = {26-46},
year = {2022},
issn = {0191-2615},
author = {Gongyuan Lu and Zili Shen and Xiaobo Liu and Yu (Marco) Nie and Zhiqiang Xiong}
}

@article{LI2022110partb,
title = {Equilibrium modeling of mixed autonomy traffic flow based on game theory},
journal = {Transportation Research Part B: Methodological},
volume = {166},
pages = {110-127},
year = {2022},
issn = {0191-2615},
author = {Jia Li and Di Chen and Michael Zhang}
}

@article{JIA2025103161partb,
title = {Adaptive signal control at partially connected intersections: A stochastic optimization model for uncertain vehicle arrival rates},
journal = {Transportation Research Part B: Methodological},
volume = {193},
pages = {103161},
year = {2025},
issn = {0191-2615},
author = {Shaocheng Jia and Sze Chun Wong and Wai Wong}
}

@article{XUE2025103209partb,
title = {Conflict-free optimal control of connected automated vehicles at unsignalized intersections: A condition-based computational framework with constrained terminal position and speed},
journal = {Transportation Research Part B: Methodological},
volume = {195},
pages = {103209},
year = {2025},
issn = {0191-2615},
author = {Yongjie Xue and Li Zhang and Yuxuan Sun and Yu Zhou and Zhiyuan Liu and Bin Yu}
}

@inproceedings{NEURIPS2024_Wang,
 author = {Wang, Junyang and Xu, Haiyang and Jia, Haitao and Zhang, Xi and Yan, Ming and Shen, Weizhou and Zhang, Ji and Huang, Fei and Sang, Jitao},
 booktitle = {Advances in Neural Information Processing Systems},
 pages = {2686--2710},
 title = {Mobile-Agent-v2: Mobile Device Operation Assistant with Effective Navigation via Multi-Agent Collaboration},
 volume = {37},
 year = {2024}
}

@inproceedings{NEURIPS2024_Ding,
 author = {Ding, Ziluo and Liu, Zeyuan and Fang, Zhirui and Su, Kefan and Zhu, Liwen and Lu, Zongqing},
 booktitle = {Advances in Neural Information Processing Systems},
 pages = {118513-118539},
 title = {Multi-Agent Coordination via Multi-Level Communication},
 volume = {37},
 year = {2024}
}

@article{ZHANG2024111796,
title = {Adaptive observer design for coupled ODE–hyperbolic PDE systems with application to traffic flow estimation},
journal = {Automatica},
volume = {167},
pages = {111796},
year = {2024},
issn = {0005-1098},
author = {Liguo Zhang and Jiahao Wu and Jingyuan Zhan}
}

@ARTICLE{Daini2025ITS,
  author={Daini, Chiara and Laura Delle Monache, Maria and Goatin, Paola and Ferrara, Antonella},
  journal={IEEE Transactions on Intelligent Transportation Systems}, 
  title={Traffic Control via Fleets of Connected and Automated Vehicles}, 
  year={2025},
  volume={26},
  number={2},
  pages={1573-1582}
}

@article{ZHOU201969partb,
title = {Distributed model predictive control approach for cooperative car-following with guaranteed local and string stability},
journal = {Transportation Research Part B: Methodological},
volume = {128},
pages = {69-86},
year = {2019},
issn = {0191-2615},
author = {Yang Zhou and Meng Wang and Soyoung Ahn}
}

@article{SASFI2023111169,
title = {Robust adaptive MPC using control contraction metrics},
journal = {Automatica},
volume = {155},
pages = {111169},
year = {2023},
issn = {0005-1098},
author = {András Sasfi and Melanie N. Zeilinger and Johannes Köhler}
}

@article{LIU2022103261,
title = {A scenario-based distributed model predictive control approach for freeway networks},
journal = {Transportation Research Part C: Emerging Technologies},
volume = {136},
pages = {103261},
year = {2022},
issn = {0968-090X},
author = {Shuai Liu and Anna Sadowska and Bart {De Schutter}}
}

@inproceedings{NEURIPS2021_Kuba,
 author = {Kuba, Jakub Grudzien and Wen, Muning and Meng, Linghui and gu, shangding and Zhang, Haifeng and Mguni, David and Wang, Jun and Yang, Yaodong},
 booktitle = {Advances in Neural Information Processing Systems},
 pages = {13458-13470},
 title = {Settling the Variance of Multi-Agent Policy Gradients},
 volume = {34},
 year = {2021}
}

@inproceedings{NEURIPS2022_Wen,
 author = {Wen, Muning and Kuba, Jakub and Lin, Runji and Zhang, Weinan and Wen, Ying and Wang, Jun and Yang, Yaodong},
 booktitle = {Advances in Neural Information Processing Systems},
 pages = {16509-16521},
 title = {Multi-Agent Reinforcement Learning is a Sequence Modeling Problem},
 volume = {35},
 year = {2022}
}

@article{JMLR-Zhong2024,
  author  = {Yifan Zhong and Jakub Grudzien Kuba and Xidong Feng and Siyi Hu and Jiaming Ji and Yaodong Yang},
  title   = {Heterogeneous-Agent Reinforcement Learning},
  journal = {Journal of Machine Learning Research},
  year    = {2024},
  volume  = {25},
  number  = {32},
  pages   = {1-67}
}

@InProceedings{Adaptive-Receding-Horizon,
author="Lukina, Anna and Esterle, Lukas and Hirsch, Christian and Bartocci, Ezio and Yang, Junxing and Tiwari, Ashish and Smolka, Scott A. and Grosu, Radu",
title="ARES: Adaptive Receding-Horizon Synthesis of Optimal Plans",
booktitle="Tools and Algorithms for the Construction and Analysis of Systems",
year="2017",
publisher="Springer Berlin Heidelberg",
address="Berlin, Heidelberg",
pages="286-302"
}
\end{document}